\documentclass[10pt,journal,twocolumn,twoside]{IEEEtran} 

\usepackage{theorem} \usepackage{cite} \usepackage{stfloats} \usepackage{epsfig} \usepackage{verbatim}
\usepackage{graphicx} \usepackage{amssymb} \usepackage{amsmath} \usepackage{color}
\usepackage{bm}
\usepackage{xcolor}
\usepackage{algorithm}
\usepackage{algorithmic}
\usepackage{makecell}

\theoremstyle{plain}
\newtheorem{theorem}{Theorem}
\theoremstyle{plain}
\newtheorem{proposition}{Proposition}
\theoremstyle{plain}
\newtheorem{corollary}{Corollary}
\theoremstyle{plain}
\newtheorem{lemma}{Lemma}
\theoremstyle{plain}
\newtheorem{assumption}{Assumption}

\newcommand\va{{\bf a}} 
\newcommand\vb{{\bf b}}

\newcommand\vf{{\bf f}}
\newcommand\vg{{\bf g}}
\newcommand\vh{{\bf h}}

\newcommand\vn{{\bf n}}

\newcommand\vp{{\bf p}}

\newcommand\vs{{\bf s}}

\newcommand\vu{{\bf u}}
\newcommand\vv{{\bf v}}

\newcommand\vy{{\bf y}}

\newcommand\mA{{\bf A}} 

\newcommand\mD{{\bf D}}

\newcommand\mF{{\bf F}}

\newcommand\mH{{\bf H}}
\newcommand\mI{{\bf I}}

\newcommand\mR{{\bf R}}
\newcommand\mS{{\bf S}}

\newcommand\mU{{\bf U}}
\newcommand\mV{{\bf V}}
\newcommand\mW{{\bf W}}
\newcommand\mX{{\bf X}}
\newcommand\mY{{\bf Y}}

\newcommand\defi{{\triangleq}}

\renewcommand{\baselinestretch}{1}

\begin{document}

\def\QEDclosed{\mbox{\rule[0pt]{1.3ex}{1.3ex}}}
\def\QEDopen{{\setlength{\fboxsep}{0pt}\setlength{\fboxrule}{0.2pt}\fbox{\rule[0pt]{0pt}{1.3ex}\rule[0pt]{1.3ex}{0pt}}}}
\def\QED{\QEDopen}
\def\proof{}
\def\endproof{\hspace*{\fill}~\QED\par\endtrivlist\unskip}

\title{Addressing the curse of mobility in massive MIMO with Prony-based angular-delay domain channel predictions}

\author{Haifan~Yin,
    Haiquan~Wang,~\IEEEmembership{Senior Member,~IEEE},
    Yingzhuang~Liu,
    and David~Gesbert,~\IEEEmembership{Fellow,~IEEE,}
\thanks{
Copyright (c) 2020 IEEE. Personal use of this material is permitted. However, permission to use this material for any other purposes must be obtained from the IEEE by sending a request to pubs-permissions@ieee.org.

A part of this work \cite{yin2020icc} was submitted to the 54th IEEE International Conference on Communications (IEEE ICC 2020).
}
\thanks{H. Yin, Y. Liu are with Huazhong University of Science and Technology, 430074 Wuhan, China (e-mail: yin@hust.edu.cn, liuyz@hust.edu.cn).
}%
\thanks{H. Wang is with Hangzhou Dianzi University, 310018 Hangzhou, China (e-mail: tx\_wang@hdu.edu.cn).
}
\thanks{D. Gesbert is with EURECOM, 06410 Biot, France (e-mail: gesbert@eurecom.fr).
}
}

\maketitle

\begin{abstract}
Massive MIMO is widely touted as an enabling technology for 5th generation (5G) mobile communications and beyond. On paper, the large excess of base station (BS) antennas promises unprecedented spectral efficiency gains. Unfortunately, during the initial phase of industrial testing, a practical challenge arose which threatens to undermine the actual deployment of massive MIMO: user mobility-induced channel Doppler. In fact, testing teams reported that in moderate-mobility scenarios, e.g., 30 km/h of user equipment (UE) speed, the performance drops up to 50\% compared to the low-mobility scenario, a problem rooted in the acute sensitivity of massive MIMO to this channel Doppler, and not foreseen by many theoretical papers on the subject.
In order to deal with this ``curse of mobility", we propose a novel form of channel prediction method, named Prony-based angular-delay domain (PAD) prediction, which is built on exploiting the specific angle-delay-Doppler structure of the multipath. In particular, our method relies on the high angular-delay resolution which arises in the context of 5G. Our theoretical analysis shows that when the number of base station antennas and the bandwidth are large, the prediction error of our PAD algorithm converges to zero for any UE velocity level, provided that only two accurate enough previous channel samples are available. Moreover, when the channel samples are inaccurate, we propose to combine the PAD algorithm with a denoising method for channel estimation phase based on the subspace structure and the long-term statistics of the channel observations.
Simulation results show that under a realistic channel model of 3GPP in rich scattering environment, our proposed method is able to overcome this challenge and even approaches the performance of stationary scenarios where the channels do not vary at all.

\end{abstract}




\section{Introduction}\label{sec_intro}

Massive multiple-input multiple-output (MIMO) introduced in \cite{marzetta:10a}, is one of the key enablers of the 5G cellular systems.
Compared to traditional MIMO with fewer base station antennas, massive MIMO can offer superior spectral efficiency and energy efficiency \cite{ngo:13} at least in theory.
One of the basic concepts is based on the fact that, as the number of BS antennas increases, the vector channel for a desired UE grows more orthogonal to the vector channel of an interfering UE, thus allowing the base station to reject interference by inexpensive precoding methods,
provided that Channel State Information (CSI) is known at base station.
CSI acquisition is known to be a formidable problem in massive MIMO. An example of CSI acquisition issue is the \emph{pilot contamination} problem.
A rich body of literature has addressed this problem. The solutions vary from angular/amplitude domain discrimination \cite{yin:13} \cite{mueller:14} \cite{yin:16}, pilot coordination \cite{yin:13}, multi-cell minimum mean square error (M-MMSE)\cite{Li:2017} \cite{Bjornson:2018}, etc.

Despite the technology hype and great expectations behind massive MIMO, some of the latest field trials have unfortunately been more than disappointing when it comes to actual system performance (see \cite{Huawei:17} \cite{FraunhoferIIS:RP-191951} for example). In particular it appeared that CSI acquisition can be severely affected in mobility scenarios. This is related to the time-varying nature of wireless channel which itself limits its coherence time, i.e., the time duration after which CSI is considered outdated. In practical cellular networks, a processing delay at the base station is inevitable because of the highly sophisticated 5G protocol, scheduling, resource allocation, encoding/decoding, and channel training under UE power constraint. This implies that even in moderate-mobility scenarios, the processing delay can end up being longer than the coherence time, making it essentially unusable for multiuser beamforming \cite{larsson:2017massive}.
More precisely, due to UE mobility, the channel can vary significantly within the time interval of CSI delay, which is defined as the duration from the time CSI is learned by the base station to the time it can be used in multiuser precoding. This channel variation, in turn, gives rise to multi-user interference if the precoder is computed based on outdated CSI.
It was for instance observed in industrial settings, that with a typical CSI delay of 4 milliseconds, the moderate-mobility scenario at 30 km/h leads to as much as 50\% of the performance reduction versus in low-mobility scenario at 3 km/h, even with relatively small number of BS antennas (e.g., 32 or 64). The performance degradation is even more severe when the number of BS antennas increases.
Solving this problem has become a priority in telecom industry. The topic of mobility enhancement has been actively discussed in 3GPP meetings recently from 5G standardization point of view. \cite{Cohere:R1-167595} proposes to address this problem by new modulation scheme named Orthogonal Time Frequency Space (OTFS), which may lead to higher diversity gains compared to Orthogonal Frequency Division Multiplexing (OFDM). \cite{FraunhoferIIS:RP-191951} and \cite{FraunhoferIIS:RP-193072} suggest that UE feeds back some information related to Doppler spectrum, which is measured based on downlink reference signals.
In academia, some information theoretic efforts of exploiting severely delayed CSI have been demonstrated but never tested in practical 5G contexts \cite{Yi:2014, Maddah-Ali:2012}.
The effects of channel aging  under a simple autoregressive (AR) model of channel time variations were studied in \cite{Truong:2013} and a linear finite impulse response (FIR) Wiener predictor was proposed. The complexity of this predictor is relatively high due to the inversion of a large matrix. The sum-rate performance with such a FIR Wiener predictor in the presence of delayed CSIT is also analyzed in \cite{Papazafeiropoulos:15} \cite{Kong:15}.  \cite{Kashyap:17} studied the performance of massive MIMO when Kalman predictor is used under a time-correlated channel aging model with rectangular spectrum. Some field trials of massive MIMO with mobility are carried out in \cite{Harris:17}. However the experiments are conducted in Line of Sight (LOS) scenarios, which simplifies the mobility problem. This is because the channel vector in LOS setting is close to a deterministic vector multiplied by a Doppler-dependent complex amplitude.

In this paper, we revisit the problem of CSI acquisition by combining it with practical and affordable
channel prediction algorithms. We propose a novel Prony-based angular-delay domain channel prediction algorithm by exploiting the structural information of the multipath channel.
More specifically, our predictor is based on the fact that the wireless channel is composed of many (e.g., several hundreds of) paths, each having a certain angle, delay, Doppler, and complex amplitude. The large number of base station antennas and the large bandwidth in 5G lead to higher resolution in both spatial and frequency domain. Our idea consists in exploiting this high resolution regime specifically. In practice the approach involves projecting the channel into an angular-delay domain, then capturing the channel variations in this domain. The intuition behind our method is to isolate one or several close-by paths from the rest, thus making the channels more predictable.
To do this, we propose to adopt here Prony's method, traditionally used in the context of spectral analysis, for its ability to predict a uniformly sampled signal composed of damped complex exponentials. In this paper we point out that this feature turns out to be useful in the 5G context because the training signal in 5G are normally periodic and the channel can be regarded as a sum of complex exponentials with each one corresponding to a path response having a Doppler component.

More specifically, the contributions of our paper are as follows:
 \begin{itemize}
   \item We first generalize the classical Prony's method to vector form and propose a vector Prony-based channel prediction algorithm, which exploits the angular-delay-Doppler structure of the wireless multipath channel to enable direct vector-domain channel prediction. Substantial gains over existing methods are observed in simulations.
   \item We propose a PAD channel prediction algorithm, which combines the high spatial and frequency resolutions of 5G massive MIMO and the angular-delay-Doppler structure of the channel. The PAD method requires less previous channel samples and achieves higher performance compared with the vector Prony prediction method. The gains over known schemes are significant.
   \item We analyze the asymptotic performance of our PAD algorithm and prove that as the number of base station antennas and the bandwidth increase, the channel prediction error converges to zero, provided that only two accurate enough channel samples are available.
   \item Finally, since in practice, current channel estimates are noisy, we improve the performance of the vector Prony method and PAD method by combining them with a denoising method using an adaptation of Tufts-Kumaresan's method \cite{Tufts:1982}.
 \end{itemize}

Simulations under the clustered delay line (CDL) channel models of 3GPP \cite{3gpp:38.901} show that our proposed method at 60 km/h of UE speed is very close to the ideal case of a stationary setting.
To the best of our knowledge, the study of channel prediction under such a realistic model of wideband massive MIMO has received little attention so far, and the high spatial-frequency resolution of 5G has not yet been fully exploited to solve the mobility challenge.

The paper is organized as follows: In Sec. \ref{modeling} we introduce the channel model of 3GPP \cite{3gpp:38.901}. In Sec. \ref{sec:dealingMobility} we first give a brief review of Prony's method, then propose the vector-based generalized Prony's method, and proceed with the proposed PAD method and its performance analysis. In Sec. \ref{sec:estimationError} we propose denoising method for the vector Prony-based algorithm and PAD algorithm. Finally, simulation results are shown in Sec. \ref{sec:numericalResult}.

Notations: We use boldface to denote matrices and vectors. Specifically, ${\mathbf{I}}$ denotes the identity matrix. ${({\mathbf{X}})^T}$, ${({\mathbf{X}})^*}$, and ${({\mathbf{X}})^H}$ denote the transpose, conjugate, and conjugate transpose of a matrix ${\mathbf{X}}$ respectively. ${({\mathbf{X}})^\dag}$ is the Moore-Penrose pseudoinverse of ${\mathbf{X}}$. $\operatorname{tr}\left\{ \cdot \right\}$ denotes the trace of a square matrix.
${\left\| \cdot \right\|_2}$ denotes the $\ell^2$ norm of a vector when the argument is a vector, and the spectral norm when the argument is a matrix.
${\left\| \cdot \right\|_F}$ stands for the Frobenius norm.
$\mathbb{E}\left\{ \cdot \right\}$ denotes the expectation.
${\bf{X}}\otimes{\bf{Y}}$ is the Kronecker product of ${\bf{X}}$ and ${\bf{Y}}$.
 $\operatorname{vec} (\mX)$ is the vectorization of the matrix $\mX$.
${\mathop{\rm diag}\nolimits} \{ {\bf{a_1,...,a_N}}\}$ denotes a diagonal matrix or a block diagonal matrix with $\bf{a_1,...,a_N}$ at the main diagonal. $\triangleq$ is used for definition. $\mathbb{N}$ and $\mathbb{N}^+$ are the set of non-negative and positive integers respectively.

\section{Channel Models}\label{modeling}
For ease of exposition, we consider an arbitrary UE in a certain cell. The antennas at the base station form a uniform planar array (UPA) with $N_v$ rows and $N_h$ columns as in commercial systems\footnote{We ignore here the polarizations in order to simplify the notations. Sec. \ref{sec:numericalResult} will incoperate the widely used dual polarized antenna model in 5G.}. Denote the number of antennas at the base station as $N_t$ and the number of antennas at the UE as $N_r$. It is clear that $N_t = N_v N_h$. The network operates in time-division duplexing (TDD) mode and the uplink (UL) and downlink (DL) occupy the same bandwidth, which consists of $N_f$ subcarriers with spacing $\triangle f$. The channel is composed of $P$ multipaths, with each path having a certain angle, delay, Doppler, and complex amplitude.

We denote the elevation departure angle, azimuth departure angle, elevation arrival angle, and azimuth arrival angle of the $p$-th path as $\theta_{p, \text{ZOD}}$, $\phi_{p, \text{AOD}}$, $\theta_{p, \text{ZOA}}$, and $\phi_{p, \text{AOA}}$ respectively. The ranges of the angles are
\begin{equation}
\theta_{p, \text{ZOD}}, \theta_{p, \text{ZOA}} \in [0, \pi],
\end{equation}
and
\begin{equation}
\phi_{p, \text{AOD}}, \phi_{p, \text{AOA}} \in (-\pi, \pi],
\end{equation}
for any $p = 1, \cdots, P$. In order to make the angular representation more rigorous, we set the azimuth angle to zero in case the elevation angle is $0$ or $\pi$, that is
\begin{equation}\label{Eq:theta0_pi}
\left\{ {\begin{array}{*{20}{c}}
{{\phi _{p,{\rm{AOD}}}} = 0,\text{ if } {\theta _{p,{\rm{ZOD}}}} = 0 \text{ or } \pi }\\
{{\phi _{p,{\rm{AOA}}}} = 0,\text{ if } {\theta _{p,{\rm{ZOA}}}} = 0 \text{ or } \pi }
\end{array}} \right.
\end{equation}

The DL channel at a certain time $t$ and a subcarrier with frequency $f$ is denoted as $\mH(f, t) \in {\mathbb{C}^{N_r \times N_t }}$. According to \cite{3gpp:38.901}, the channel between the $s$-th base station antenna and the $u$-th UE antenna is modeled as
\begin{align}\label{Eq:hft}
\mbox{\small$\displaystyle
{{h}}_{u,s}(f, t) = \sum\limits_{p = 1}^P {\beta _{p}} e^{\frac{j2\pi \hat{r}_{\text{rx},p}^T \bar{d}_{\text{rx}, u}}{\lambda_0}} e^{\frac{j2\pi \hat{r}_{\text{tx},p}^T \bar{d}_{\text{tx}, s}}{\lambda_0}} {e^{ - j2\pi f{\tau _p}}} e^{j\omega_p t},$}
\end{align}
where $\beta _{p}$ and $\tau_p$ are the complex amplitude and the delay of the $p$-th path respectively. $\lambda_0$ is the wavelength of center frequency. $\hat{r}_{\text{rx},p}$ is the spherical unit vector with azimuth arrival angle $\phi_{p, \text{AOA}}$ and elevation arrival angle $\theta_{p, \text{ZOA}}$:
\begin{equation}\label{Eq:rhat_rx_p}
\hat{r}_{\text{rx},p} \triangleq \left[ {\begin{array}{*{20}{c}}
{\sin {\theta_{p, \text{ZOA}}}\cos {\phi_{p, \text{AOA}}}}\\
{\sin {\theta_{p, \text{ZOA}}}\sin {\phi_{p, \text{AOA}}}}\\
{\cos {\theta_{p, \text{ZOA}}}}
\end{array}} \right].
\end{equation}
Likewise, $\hat{r}_{\text{tx},p}$ is the spherical unit vector defined as:
\begin{equation}\label{Eq:rhat_tx_p}
\hat{r}_{\text{tx},p} \triangleq \left[ {\begin{array}{*{20}{c}}
{\sin {\theta_{p, \text{ZOD}}}\cos {\phi_{p, \text{AOD}}}}\\
{\sin {\theta_{p, \text{ZOD}}}\sin {\phi_{p, \text{AOD}}}}\\
{\cos {\theta_{p, \text{ZOD}}}}
\end{array}} \right].
\end{equation}
$\bar{d}_{\text{rx}, u}$ is the $u$-th UE antenna's location vector which contains the 3D cartesian coordinate. Similarly, $\bar{d}_{\text{tx}, s}$ is the location vector of the $s$-th base station antenna. The last exponential term $e^{j\omega_p t}$ is the Doppler of the $p$-th path, where $t$ denotes time. $\omega_p$ is defined as $\omega_p \triangleq {\hat{r}_{\text{rx},p}^T \bar{v} }/{\lambda_0}$,
where $\bar{v}$ is the UE velocity vector with speed $v$, travel azimuth angle $\phi_v$, and travel elevation angle $\theta_v$:
\begin{equation}\label{Eq:v_bar}
\bar{v} \triangleq v {[\begin{array}{*{20}{c}}
{\sin {\theta _v}\cos {\phi _v}}&{\sin {\theta _v}\sin {\phi _v}}&{\cos {\theta _v}}
\end{array}]^T}.
\end{equation}
An illustration of the coordinate system
is shown in Fig. \ref{fig:3Dmodel}.
\begin{figure}[h]
  \centering
  \includegraphics[width=3.5in]{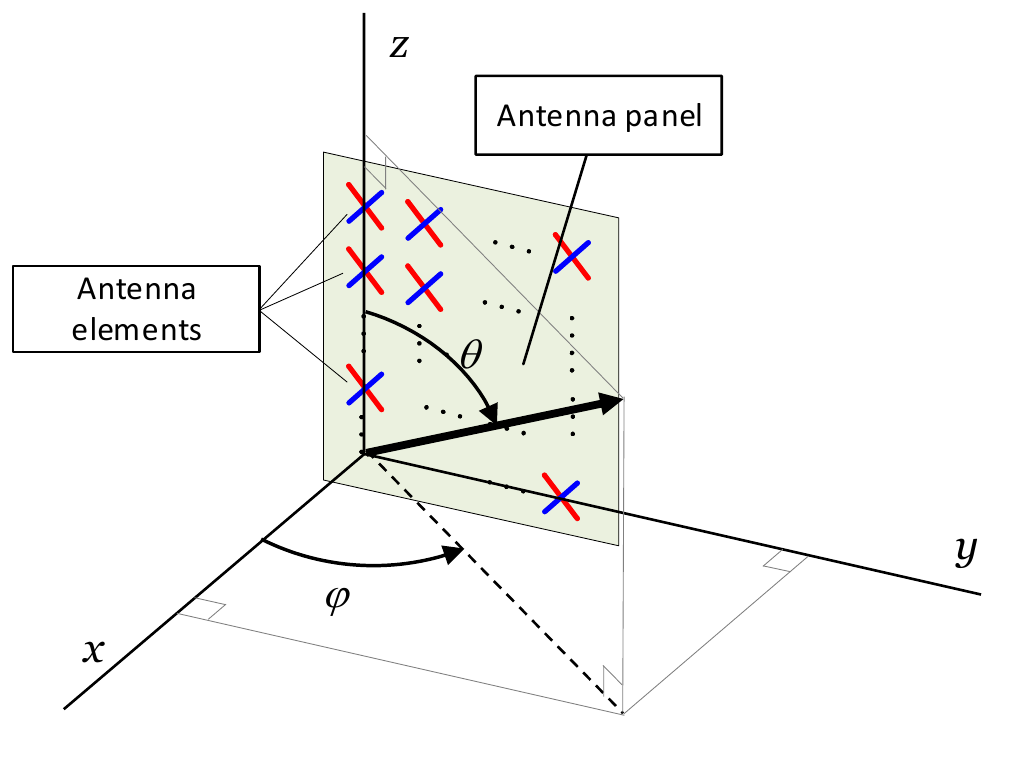}\\
  \caption{Definition of the coordinate system.} \label{fig:3Dmodel}
\end{figure}

Note that shifting or rotating the coordinate system has little impact on the channel model. Without loss of generality, we let the origin be at the first base station antenna which is located at the lower left corner of the antenna panel, as shown in Fig. \ref{fig:3Dmodel}. The antenna panel is on YZ plane. The antenna index starts from the lower left corner of the antenna panel and increases along the Z-axis until the top row, then continues with the second column, third column, etc.
Define the 3-D steering vector of a certain path with elevation departure angle ${\theta}$ and azimuth departure angle ${\phi}$ as
\begin{equation}
{\bf{a}}({\theta},{\phi}) = {{\bf{a}}_h}({\theta},{\phi}) \otimes {{\bf{a}}_v}({\theta}),
\end{equation}
where
\begin{small}
\begin{align}
& {{\bf{a}}_h}(\theta ,\phi ) = \\ \nonumber
& {\left[ {\begin{array}{*{20}{c}}
1&{{e^{j2\pi \frac{{{D_h}\sin (\theta )\sin (\phi )}}{{{\lambda _0}}}}}}& \cdots &{{e^{j2\pi \frac{{({N_h} - 1){D_h}\sin (\theta )\sin (\phi )}}{{{\lambda _0}}}}}}
\end{array}} \right]^T},
\end{align}
\end{small}
and
\begin{equation}
\mbox{\small$\displaystyle
{{\bf{a}}_v}(\theta ) = {\left[ {\begin{array}{*{20}{c}}
1&{{e^{j2\pi \frac{{{D_v}\cos (\theta )}}{{{\lambda _0}}}}}}& \cdots &{{e^{j2\pi \frac{{({N_v} - 1){D_v}\cos (\theta )}}{{{\lambda _0}}}}}}
\end{array}} \right]^T},$}
\end{equation}
with ${D_h}$ and ${D_v}$ being the horizontal and vertical antenna spacing at the base station respectively.

Let $\vh_u(f, t) \in {\mathbb{C}^{1 \times N_t }}$ denote the channel between all base station antenna and the $u$-th UE antenna at time $t$ and frequency $f$. We write the channels at all $N_f$ subcarriers in a matrix form:
\begin{equation}\label{Eq:Hut}
{{\bf{H}}_u}(t) \triangleq [\begin{array}{*{20}{c}}
{{\bf{h}}_u^T({f_1},t)}&{{\bf{h}}_u^T({f_2},t)}& \cdots &{{\bf{h}}_u^T({f_{{N_f}}},t)}
\end{array}],
\end{equation}
where $f_i$ is the frequency of the $i$-th $(1 \leq i \leq N_f)$ subcarrier.
According to the model in Eq. (\ref{Eq:hft}), we may further write
\begin{equation}\label{Eq:HutACB}
{{\bf{H}}_u}(t) = {\bf{A}}{{\bf{C}}_u}(t){{\bf{B}}},
\end{equation}
where ${\bf{A}} \in {\mathbb{C}^{N_t \times P}}$ is composed of $P$ 3-D steering vectors:
\begin{small}
\begin{equation}
{\bf{A}}{\rm{ }} \buildrel \Delta \over = \left[ {\begin{array}{*{20}{c}}
{{\bf{a}}({\theta _{1,{\rm{ZOD}}}},{\phi _{1,{\rm{AOD}}}})}& \cdots &{{\bf{a}}({\theta _{P,{\rm{ZOD}}}},{\phi _{P,{\rm{AOD}}}})}
\end{array}} \right],
\end{equation}
\end{small}
and
\begin{equation}
{\bf{B}} \buildrel \Delta \over = \left[ {\begin{array}{*{20}{c}}
{{\bf{b}}({\tau _1})}&{{\bf{b}}({\tau _2})}& \cdots &{{\bf{b}}({\tau _P})}
\end{array}} \right]^T,
\end{equation}
with ${{\bf{b}}({\tau _p})}$, $(p = 1, \cdots, P)$ being the delay response vector of the $p$-th path, which is defined as
\begin{equation}
\mbox{\small$\displaystyle
{\bf{b}}({\tau _p}) = {\left[ {\begin{array}{*{20}{c}}
{{e^{ -j2\pi {f_1}{\tau _p}}}}&{{e^{ -j2\pi {f_2}{\tau _p}}}}& \cdots &{{e^{ -j2\pi {f_{{N_f}}}{\tau _p}}}}
\end{array}} \right]^T},$}
\end{equation}
and ${{\bf{C}}_u}(t) =  {\mathop{\rm diag}\nolimits} \{ {{c_{u,1}}(t),...,{c_{u,P}}(t)}\} \in \mathbb{C}^{P \times P}$ is a diagonal matrix with its $p$-th ($p = 1, \cdots, P$) diagonal entry being
\begin{equation}
{c_{u,p}}(t) \buildrel \Delta \over = {\beta _p}{e^{\frac{{j2\pi \hat r_{{\rm{rx}},p}^T{{\bar d}_{{\rm{rx}},u}}}}{{{\lambda _0}}}}}{e^{j{\omega _p}t}}.
\end{equation}

The vectorized form of Eq. (\ref{Eq:HutACB}) is given by
\begin{equation}\label{Eq:hbaru}
\bm{\hbar}_u(t) = \operatorname{vec} ({{\bf{H}}_u}(t)) = \sum\limits_{p = 1}^P {c_{u,p}}(t) \vv_p,
\end{equation}
where
\begin{equation}\label{Eq:vp}
\vv_p = {\bf{b}}({\tau _p}) \otimes {{\bf{a}}({\theta _{p,{\rm{ZOD}}}},{\phi _{p,{\rm{AOD}}}})}.
\end{equation}
$\vv_p$ is a generalized steering vector, which reflects the angle and delay response of the $p$-th path in a wideband multiple antenna system.
From Eq. (\ref{Eq:HutACB}) or Eq. (\ref{Eq:hbaru}) we have the observation that
the channel is highly structural in both spatial and frequency domain. Each path is associated with a certain steering vector and delay response vector, depending on its angle and delay. This structural information is hidden in the generalized steering vectors.

\section{Dealing with mobility through predictions}\label{sec:dealingMobility}

\subsection{The challenge of mobility in massive MIMO}\label{sec:mobilityProblem}
As is well known, channel time variability can create inter-user interference induced by a precoder which is computed based on aging CSI.
This impediment can be mitigated by anticipating the future channel variations. While predicting the future fading state of a wireless channel is a very challenging task, the accounting of the specific space-time structure of the channel which arises in a broadband context (as in 5G) opens fresh perspectives for improvement.

\subsection{A review of Prony's method}\label{sec:pronyReview}
Prony's method proposed by Gaspard Riche de Prony in 1795 is a useful tool to analyze a uniformly sampled signal composed of a number of damped complex exponentials \cite{Prony:1795} and extract valuable information (e.g., the amplitudes and frequencies of the exponentials) which can be used for prediction. A review of this method is briefly given below.
Suppose we have $K$ samples of data $y(k)$ which consist of $N$ exponentially damped signals:
\begin{equation}\label{Eq:yk}
y(k) = \sum\limits_{n = 1}^N {{\beta _n}{e^{({-\alpha _n} + j{2\pi f_n})k}}} ,0 \le k \le K - 1,
\end{equation}
where $\alpha _n$ (positive) and $f_n$ ($1 \le n \le N$) are the pole damping factor and pole frequency respectively. $\beta _n$ ($1 \le n \le N$) is the complex amplitude. Note that in the special case of channel prediction, $y(k)$ can be regarded as the uniformly sampled channel estimate.
Define the following polynomial:
\begin{equation}\label{Eq:P0}
{P_0}(z) \defi \prod\limits_{n = 1}^N {\left( {z - {e^{{s_n}}}} \right)} = \sum\limits_{n = 0}^N p_n z^n, z \in {\mathbb{C}},
\end{equation}
where $s_n = {-\alpha _n} + j{2\pi f_n}$ for $n = 1, \cdots, N$. It is clear that $p_N = 1$ and $e^{s_n}, (n = 1, \cdots, N)$ are zeros of ${P_0}(z)$.
For an arbitrary $m \in \mathbb{N}$, one has
\begin{align}\label{Eq:pny}
\sum\limits_{n = 0}^N {{p_n}y(n + m)} & = \sum\limits_{n = 0}^N {{p_n}\sum\limits_{l = 1}^N {{\beta _l}{e^{{s_l}(n + m)}}} } \nonumber \\
& =  \sum\limits_{l = 1}^N {{\beta _l}{e^{{s_l}m}}\left( {\sum\limits_{n = 0}^N {{p_n}{e^{{s_l}n}}} } \right)}  \overset{{a}}{=} 0,
\end{align}
where $\overset{{a}}{=}$ is due to the fact that $e^{s_l} (l = 1, \cdots, N)$ are zeros of ${P_0}(z)$. Eq.(\ref{Eq:pny}) implies that the following homogeneous linear difference equation is fulfilled:
\begin{equation}
\sum\limits_{n = 0}^{N - 1} {{p_n}y(n + m)}  =  - y(N + m),m \in \mathbb{N}.
\end{equation}

Thus, we may obtain the coefficients $p_n$ with the $2N$ sampled data by solving
the following linear equations:
\begin{equation}\label{Eq:Yp}
\mY \vp = - \vh,
\end{equation}
where $\mY$ is a square Hankel matrix
\begin{small}
\begin{equation}\label{Eq:Y}
\mY \triangleq \left[ {\begin{array}{*{20}{c}}
{y(0)}&{y(1)}& \cdots &{y(N - 1)}\\
{y(1)}&{y(2)}& \cdots &{y(N)}\\
 \vdots & \vdots & \vdots & \vdots \\
{y(N-1)}&{y(N)}& \cdots &{y(2N - 2)}
\end{array}} \right],
\end{equation}
\end{small}
\begin{equation}\label{Eq:p}
{\bf{p}} \triangleq {[\begin{array}{*{20}{c}}
{{p_0}}&{{p_1}}& \cdots &{{p_{N - 1}}}
\end{array}]^T},
\end{equation}
\begin{equation}\label{Eq:h}
{\bf{h}} \triangleq {[\begin{array}{*{20}{c}}
{y(N)}&{y(N + 1)}& \cdots &{y(2N - 1)}
\end{array}]^T}.
\end{equation}
The least squares solution to Eq. (\ref{Eq:Yp}) is given by
\begin{equation}\label{Eq:phat}
\hat {\bf{p}} = \mathop {\arg \min }\limits_{\bf{p}} {\left\| {{\bf{Yp}} + {\bf{h}}} \right\|^2} = - {\bf{Y}^{\dag}} \vh.
\end{equation}
Note that we may need $K \geq 2N$ samples to obtain all the coefficients $p_n$, $n = 0, \cdots, N-1$.

\subsection{Channel prediction based on vector Prony method}
While Prony's method is presented in scalar form in the literature, we propose an extension to vector form in this paper.
Consider a uniformly sampled signal vector composed of weighted sum of constant vectors where the weights are damped complex exponentials. Suppose we have $K$ samples of signal vector $\vy(k), k = 0, \cdots, K-1$:
\begin{equation}\label{Eq:vyk}
{\bf{y}}(k) = \sum\limits_{n = 1}^N {{{\bf{a}}_n}{e^{s_n k}}}  = \mA{\left[ {\begin{array}{*{20}{c}}
{{e^{{s_1}k}}}& \cdots &{{e^{{s_N}k}}}
\end{array}} \right]^T},
\end{equation}
where ${\bf{A}} \defi \left[ {\begin{array}{*{20}{c}}
{{{\bf{a}}_1}}&{{{\bf{a}}_2}}& \cdots &{{{\bf{a}}_N}}
\end{array}} \right]$; $s_n \defi -\alpha _n + j{2\pi f_n}$ and $\va_n \in \mathbb{C} ^{M \times 1}$ is a time-invariant vector for $n = 1, \cdots, N$.
In the context of channel prediction, ${\bf{a}}_n$ can be a steering vector or a generalized steering vector.

We use the same polynomial $P_0(z)$ in Eq. (\ref{Eq:P0}) with $e^{s_n}, (n = 1, \cdots, N)$ being zeros and $p_0, p_1, \cdots, p_N$ being the coefficients.
For $\forall m \in \mathbb{N}$, we have
\begin{small}
\begin{align*}
& \sum\limits_{n = 0}^N {{p_n}{\bf{y}}(n + m) = } \sum\limits_{n = 0}^N {{p_n}{\bf{A}}{{\left[ {\begin{array}{*{20}{c}}
{{e^{{s_1}(n + m)}}}& \cdots &{{e^{{s_N}(n + m)}}}
\end{array}} \right]}^T}} \nonumber \\
= & {\bf{A}}\left[ {\begin{array}{*{20}{c}}
{\sum\limits_{n = 0}^N {{p_n}{e^{{s_1}(n + m)}}} }\\
{\sum\limits_{n = 0}^N {{p_n}{e^{{s_2}(n + m)}}} }\\
 \vdots \\
{\sum\limits_{n = 0}^N {{p_n}{e^{{s_N}(n + m)}}} }
\end{array}} \right] = {\bf{A}}\left[ {\begin{array}{*{20}{c}}
{{e^{{s_1}m}}\sum\limits_{n = 0}^N {{p_n}{e^{{s_1}n}}} }\\
{{e^{{s_2}m}}\sum\limits_{n = 0}^N {{p_n}{e^{{s_2}n}}} }\\
 \vdots \\
{{e^{{s_N}m}}\sum\limits_{n = 0}^N {{p_n}{e^{{s_N}n}}} }
\end{array}} \right] = {\bf{0}}.
\end{align*}
\end{small}
Thus, we can compute the coefficients of $p_0, p_1, \cdots, p_{N-1}$ by solving the following linear equations
\begin{equation}\label{Eq:mYp}
 \mY \vp =  - {\bf{y}}(N),
\end{equation}
where $\mY \defi \left[ {\begin{array}{*{20}{c}}
{{\bf{y}}(0)}&{{\bf{y}}(1)}& \cdots &{{\bf{y}}(N - 1)}
\end{array}} \right]$ and $\vp \triangleq {[\begin{array}{*{20}{c}}
{{p_0}}&{{p_1}}& \cdots &{{p_{N - 1}}}
\end{array}]^T}$.
The least squares estimate of the coefficient is $\hat{\vp} = - \mY^\dag \vy(N)$.

We now apply this method to channel predictions. Denote the vectorized channel of the whole bandwidth at time $t$ as $\underline{\bm{\hbar}}(t) \in \mathbb{C}^{N_t N_r N_f \times 1}$:
\begin{equation}\label{Eq:hwbvec}
\underline{\bm{\hbar}}(t) \defi {[\begin{array}{*{20}{c}}
\bm{\hbar}_1(t)^T &\bm{\hbar}_2(t)^T & \cdots &\bm{\hbar}_{N_r}(t)^T
\end{array}]^T},
\end{equation}
where $\bm{\hbar}_u(t), u = 1, \cdots, N_r$, is defined in Eq. (\ref{Eq:hbaru}).
Our target is to overcome the CSI delay by channel prediction based on previous samples which are equally spaced in time. In practice, the samples are obtained using periodic Sounding Reference Signal (SRS) transmitted by a UE. The period of SRS $\bigtriangleup T$ can be as short as one slot (14 OFDM symbols for normal cyclic prefix case) \cite{3gpp:38.331}. Taking the 15 kHz subcarrier spacing for example, the minimal period is 1 ms. In fact for other configurations of subcarrier spacing, e.g., 30 kHz or more, the period can be much shorter.
We assume the CSI delay $T_d = N_d \triangle T, N_d \in \mathbb{N}$. Denote the known samples as $\underline{\bm{\hbar}}(t_0), \underline{\bm{\hbar}}(t_1), \cdots, \underline{\bm{\hbar}}(t_L)$. Let $\bm{\mathcal{H}} \defi [\underline{\bm{\hbar}}(t_0), \underline{\bm{\hbar}}(t_1), \cdots, \underline{\bm{\hbar}}(t_{L-1})]$.
One may have the following prediction algorithm of order $N = L$ to predict the channel of $T_d$ time later.

\begin{algorithm}
\caption{Vector Prony-based channel prediction}
\begin{algorithmic}[1]\label{Alg:PronyVector}
\STATE{Compute the least squares estimate of the Prony coefficients $\hat \vp = -\bm{\mathcal{H}}^{\dag} \underline{\bm{\hbar}}(t_{L})$;}

\STATE{Update $\bm{\mathcal{H}} \leftarrow \left[ \underline{\bm{\hbar}}(t_{1}), \cdots ,\underline{\bm{\hbar}}(t_{L}) \right]$;}

\STATE{Compute the channel prediction at $t_{L+1}$, $\hat{\underline{\bm{\hbar}}}(t_{L+1}) = -\bm{\mathcal{H}}\hat \vp $;}

\STATE{\textbf{for} $i = 2, \cdots, N_d$}

\STATE{\quad Update $\bm{\mathcal{H}}$ by removing its first column and appending the previously predicted channel to its last column: $\bm{\mathcal{H}} \leftarrow \left[ {\underline{\bm{\hbar}}({t_i}), \cdots ,\underline{\hat{\bm{\hbar}}}({t_{L + i - 1}})} \right]$;}

\STATE{\quad Compute $\hat{\underline{\bm{\hbar}}}(t_{L+i}) = -\bm{\mathcal{H}} \hat{\vp}$;}

\STATE{\textbf{end for}}

\end{algorithmic}
\end{algorithm}
Note that in case $N_d = 1$, step 4 - step 6 are not needed. The minus sign in step 1, step 3, and step 6 can be all removed without affecting the results. In fact, we choose to predict each time the whole wideband channel so that only one matrix inversion (of size $N \times N$) is needed, which helps to reduce the computational complexity. Other possibilities include predicting each time for a certain subcarrier or for a certain UE antenna $u$, e.g., $\vh_u(f, t)$, however at the expense of more $N \times N$ matrix inversions. The complexity of this algorithm is dominated by the computation of the Prony coefficients $\hat \vp$ and the prediction of $\hat{\underline{\bm{\hbar}}}(t_{L+i})$. It can be verified that vector Prony-based algorithm has a complexity order of $\mathcal{O}(N^2 N_t N_r N_f) + \mathcal{O}(N_d N N_t N_r N_f)$.

\subsection{Prony-based angular-delay domain channel prediction}
As shown in Eq. (\ref{Eq:hft}), the channel is composed of $P$ paths, and each path has a Doppler term $e^{j \omega_p t}, p = 1, \cdots, P$. The number of paths can be large, which makes the prediction accuracy degrade if only a limited number of samples are available. In order to cope with this problem, we propose a Prony-based angular-delay domain (PAD) channel prediction method. The main idea is to convert the channel into another domain where the Doppler terms of different paths are less intertwined with each other. We choose this domain in such a way that it reflects the geometry of the antenna array and the wideband delay response structure of the channel. As indicated by \cite{yin:13} and \cite{adhikary:13}, the steering vectors of a uniform linear array (ULA) can be well approximated by the columns of discrete Fourier transform (DFT) matrix as the number of antennas increases. Here we apply this result in the 3-D steering vector case for the UPA array. Denote a DFT matrix of size $K \times K$ as
\begin{align*}
\mbox{\small$\displaystyle
\mW(K) \triangleq \frac{1}{\sqrt{K}}\begin{bmatrix}
\omega^{0\cdot0} & \omega^{0\cdot1} & \cdots & \omega^{0(K-1)} \\
\omega^{1\cdot0} & \omega^{1\cdot1} & \cdots & \omega^{1(K-1)} \\
\vdots & \vdots & \ddots & \vdots \\
\omega^{(K-1)\cdot0} & \omega^{(K-1)\cdot 1}& \cdots & \omega^{(K-1)(K-1)} \end{bmatrix},$}
\end{align*}
where $\omega \triangleq e^{-2 \pi j / K}$. Since UPA antenna array is considered, a DFT-based spatial orthogonal basis can be obtained as $\mW(N_h) \otimes \mW(N_v)$,
where $N_h$ and $N_v$ are the number of columns and the number of rows of antennas on the UPA respectively. Thus each column of $\mW(N_h) \otimes \mW(N_v)$ can be regarded as a spatial beam that reflects the array topology shown in Fig. \ref{fig:3Dmodel}.
Likewise, the frequency orthogonal basis is $\mW(N_f)$. Notice that in practice, $N_f$ can also denote the number of resource blocks (RBs) or the number of groups of consecutive RBs  depending on the SRS frequency structure. In such cases, $N_f$ is much smaller than the total number of subcarriers.
The joint spatial-frequency orthogonal basis can be computed as\footnote{Another option of the orthogonal basis is $\mS \defi \mW(N_f) \otimes \mW(N_h)^H \otimes \mW(N_v)^H$, which seems more consistent with the representations of the steering vector and delay response vector. However this substitution has no impact on the performance.}
\begin{equation}\label{Eq:Smat}
\mS \defi \mW(N_f) \otimes \mW(N_h) \otimes \mW(N_v).
\end{equation}
We project the vectorized channel Eq. (\ref{Eq:hbaru}) onto the spatial-frequency orthogonal basis $\mS$.
\begin{equation}\label{Eq:Gt}
\vg_u(t) \defi \mS^H \bm{\hbar}_u(t).
\end{equation}
$\vg_u(t) \in \mathbb{C}^{{N_t N_f} \times 1}$ is in fact the vectorized representation of the channel in angular-delay domain.
Due to finite number of multipath components, most of the elements in vector $\vg_u(t)$ are close to zero when the number of BS antennas and the bandwidth are large. As a result we may ignore the insignificant elements in $\vg_u(t)$ and focus on the predictions of the significant ones.
Let ${\tilde{\bf{g}}_u}({t_l})$ be the re-arranged ${{\bf{g}}_u}({t_l})$ with its absolute values in non-increasing order.
The number of non-negligible angular-frequency positions $N_s$ is defined as
\begin{equation*}
\mbox{\small$\displaystyle
N_s \defi \mathop {\arg \min }\limits_{{N_s}} \left\{ {\sum\limits_{u = 1}^{{N_r}} \sum\limits_{l = 0}^L {\sum\limits_{n = 1}^{{N_s}} {{{\left| {{\tilde{\bf{g}}_u}({t_l},n)} \right|}^2}} }  \ge \gamma \sum\limits_{u = 1}^{N_r} \sum\limits_{l = 0}^L {| {{\vg_u}({t_l})} |^2} } \right\},$}
\end{equation*}
where ${{\tilde{\bf{g}}_u}({t_l},n)}$ is the $n$-th entry of ${\tilde{\bf{g}}_u}({t_l})$, $\gamma$ is a positive threshold that is close to 1. The physical meaning of $\gamma$ is the ratio between the sum power of non-negligible elements and the total power of the channel. Note that $N_s$ is normally much smaller than the size of the vector $\vg_u(t)$: $N_s \ll N_t N_f$. Thus by ignoring the insignificant elements, we may greatly reduce the computational complexity in channel prediction.
We use $g_{u,n}(t), (n = 1, \cdots, N_s)$ to denote the $n$-th non-negligible entry, which is located at the $r(n)$-th row of the vector $\vg_u(t)$.
The vectorized channel can be approximated as
\begin{equation}\label{Eq:Hubar}
{\bm{\hbar}}_u(t) \approx \sum\limits_{n = 1}^{{N_s}} {{g_{u,n}(t)}{\bf{s}}_{r(n)}},
\end{equation}
where ${{\bf{s}}_{i}}$ is the $i$-th column of $\mS$.
We seek to predict the channel at each of the $N_s$ angle-delay pairs using Prony's method with $L+1$ samples ${{g}_{u,n}(t_0)}, \cdots, {{g}_{u,n}(t_L)}$. Without loss of generality, we assume $L$ is odd and let the order of the predictor $N = (L+1)/2$.
For a certain $n, 1 \le n \le N_s$,  we may obtain the Prony coefficients by solving
\begin{equation}\label{Eq:Gp}
{\bm{\mathcal G}}(u,n) \vp(u,n) = - \vg(u,n),
\end{equation}
where ${\bm{\mathcal G}}(u,n) \triangleq $
\begin{equation}\label{Eq:mathcalG}
\left[ {\begin{array}{*{20}{c}}
{{g_{u,n}}({t_0})}&{{g_{u,n}}({t_1})}& \cdots &{{g_{u,n}}({t_{{{\rm{N-1}}}}})}\\
{{g_{u,n}}({t_1})}&{{g_{u,n}}({t_2})}& \cdots &{{g_{u,n}}({t_{\rm{N}}})}\\
 \vdots & \vdots & \vdots & \vdots \\
{{g_{u,n}}({t_{\rm{N-1}}})}&{g_{u,n}}{(t_N)}& \cdots &{{g_{u,n}}({t_{\rm{2N-2}}})}
\end{array}} \right]
\end{equation}
\begin{equation}
{{\bf{p}}(u,n)} \triangleq {[\begin{array}{*{20}{c}}
{{p_0(u,n)}}& \cdots &{{p_{N - 1}(u,n)}}
\end{array}]^T},
\end{equation}
\begin{equation*}
{{\bf{g}}(u,n)} \triangleq {[\begin{array}{*{20}{c}}
{{g}_{u,n}(t_{\rm{N}})}&{{g}_{u,n}(t_{\rm{N+1}})}& \cdots &{{{g}_{u,n}(t_{\rm{2N-1}})}}
\end{array}]^T}.
\end{equation*}
The least squares estimate of ${{\bf{p}}(u,n)}$ is
\begin{equation}\label{Eq:phat_un}
{\hat{\bf{p}}(u,n)} = -{\bm{\mathcal G}}^{\dag}(u,n) {{\bf{g}}(u,n)}.
\end{equation}
The prediction of ${g}_{(u,n)}(t_{\rm{L+1}})$ is given by
\begin{equation}\label{Eq:gPred}
\hat{g}_{u,n}(t_{\rm{L+1}}) = - \vg(u,n,L) {\hat{\bf{p}}(u,n)},
\end{equation}
where
\begin{align}
\vg(u,n,L) \defi \left[ {\begin{array}{*{20}{c}}{{g_{u,n}}({t_{\rm{L-N+1}}})}& \cdots &{{g_{u,n}}({t_{{{\rm{L}}}}})}\end{array}} \right].
\end{align}

When $N_d > 1$, we may repeat computing Eq. (\ref{Eq:gPred}) $N_d - 1$ times and update $\vg(u,n,L)$ each time by removing the first column and appending the previous predict to the last column, until we obtain the prediction $\hat{g}_{u,n}(t_{\rm{L+N_d}})$.

The method is summarized in Algorithm \ref{Alg:PronyAngleDelay}.
\begin{algorithm}
\caption{PAD channel prediction method}
\begin{algorithmic}[1]\label{Alg:PronyAngleDelay}
\STATE{Compute the angular-delay domain channel $\vg_u(t_l)$ for $l = 0, \cdots, L$ and $u = 1, \cdots, N_r$ according to Eq. (\ref{Eq:Gt}).}

\STATE{Find the non-negligible values $g_{u,n}(t_l)$ and their positions $r(n)$ for $u = 1, \cdots, N_r, n = 1, \cdots, N_s, l = 0, \cdots, L$.}

\STATE{\textbf{for} $u = 1, \cdots, N_r$}

\STATE{\quad \textbf{for} $n = 1, \cdots, N_s$}

\STATE{\quad Compute the least squares estimate of the Prony coefficients as in Eq. (\ref{Eq:phat_un});}

\STATE{\quad Repeat Eq. (\ref{Eq:gPred}) $N_d$ times to compute the prediction $\hat{g}_{u,n}(t_{\rm{L+N_d}})$;}

\STATE{\quad \textbf{end for}}

\STATE{\quad Reconstruct the channel prediction at $t_{\rm{L+N_d}}$ as in Eq. (\ref{Eq:Hubar}) with ${g_{u,n}(t)}$ replaced by $\hat{g}_{u,n}(t_{\rm{L+N_d}})$;}

\STATE{\textbf{end for}}

\end{algorithmic}
\end{algorithm}
Note that we show in this section a DFT based angular-frequency domain projection as it is simple to implement in practice. In fact we may also adopt other angle and delay estimation methods, e.g., Multiple Signal Classification (MUSIC) \cite{R.schmidt:1986music}, Estimation of Signal Parameters via Rational Invariance Techniques (ESPRIT) \cite{roy:1989esprit}, etc.  However these advanced methods generally entails relatively high complexity due to a multi-dimensional search. The complexity of our PAD algorithm is now analyzed. The spatial-frequency orthogonal projection can be effectively computed by Fast Fourier Transform (FFT) algorithms, which has a complexity of $2N N_t N_f \log (N_t N_f)$.
The Prony coefficient computation has a complexity order of $\mathcal{O}(N_s N^{2.37} )$ due to the matrix inversion. The prediction of the channel at $t_{\rm{L+i}}, i = 1, \cdots, N_d$ has a complexity of $N_d N_s N$. The complexity of channel reconstruction using Eq. (\ref{Eq:Hubar}) is no greater than the orthogonal projection. As a result, the complexity order is $\mathcal{O}(N N_t N_f \log (N_t N_f)) + \mathcal{O}(N_s N^{2.37} ) + \mathcal{O}(N_d N_s N)$. Moreover, the order of the predictor $N$ will converge to 1 as the number of antennas and the bandwith increase, which will be shown in the subsequent section of performance analysis. Thus the complexity order of our PAD algorithm will converge to $\mathcal{O}(N_t N_f \log (N_t N_f))$ in massive MIMO and large bandwidth regime.

\subsection{Performance analysis of the PAD algorithm}\label{sec:performAnalysisAlgo2}
The asymptotical performance of our PAD algorithm is now analyzed. Define a tuple $(\theta _{p,{\rm{ZOD}}}, {\phi _{p,{\rm{AOD}}}}, \tau_p)$ which contains the elevation/azimuth departure angle and delay of the $p$-th path. Regarding the tuple, we let the equal sign $=$ denote the case when two tuples are completely equal. In other words, $(\theta _{p,{\rm{ZOD}}}, {\phi _{p,{\rm{AOD}}}}, \tau_p) = (\theta _{q,{\rm{ZOD}}}, {\phi _{q,{\rm{AOD}}}}, \tau_q)$ if and only if $\theta _{p,{\rm{ZOD}}} = \theta _{q,{\rm{ZOD}}}, {\phi _{p,{\rm{AOD}}}} = {\phi _{q,{\rm{AOD}}}}, \text{and } \tau_p = \tau_q$. $(\theta _{p,{\rm{ZOD}}}, {\phi _{p,{\rm{AOD}}}}, \tau_p) \neq (\theta _{q,{\rm{ZOD}}}, {\phi _{q,{\rm{AOD}}}}, \tau_q)$ means one or more entries in one tuple are not equal to the corresponding entries in the other.
Here we define a stationary time concept - the time duration over which the multipath angles and delays are stationary. We build our analysis under the assumption that the stationary time is no shorter than the CSI delay. This is in general a reasonable assumption, which is also validated in real-world measurement \cite{FraunhoferIIS:RP-193072}. Consider a vehicle moving at 100 km/h of speed. Within a CSI delay of 4 millisecond, the vehicle moves only around 10 centimeters. Such a small displacement will not introduce substantial changes in angles and delays of the multipath in a macro-cell environment.

Before stating our main result in Theorem  \ref{theoNoiseFreeEstimationErr}, we introduce three intermediate lemmas.
For notational simplicity, we mostly drop the subscripts of ``ZOD" and ``AOD" in the lemmas and Theorem \ref{theoNoiseFreeEstimationErr} and their corresponding proofs. Note that throughout the paper we make the implicit and realistic assumption that the delay and UE velocity level are finite. We also introduce here a mild technical assumption.
\begin{assumption}\label{assumption1}
For any two paths, $p \neq q$, at least one of the following three attributes are different from one another: the elevation departure angle, the azimuth departure angle, and the delay:
\begin{align}\label{Eq:noiseFreeCond}
& (\theta _{p}, {\phi _{p}}, \tau_p) \neq (\theta _{q}, {\phi _q}, \tau_q), \\
& \forall p,q = 1, \cdots, P \text{ and } p \neq q.
\end{align}
\end{assumption}
Remarks:
This assumption is in general valid. We will show more results in Corollary \ref{coro:NcPathCommonTuple} and Corollary \ref{coro:NcPathCommonTupleNB} when this assumption is not true.
More clarifications of Assumption \ref{assumption1} is shown in Lemma \ref{lemma:mutualOthogonal}.

\begin{lemma}\label{lemma:mutualOthogonal}
The generalized steering vectors $\vv_p$ and $\vv_q$ are asymptotically orthogonal:
\begin{equation}\label{Eq:vpvq}
\lim_{N_v, N_h, N_f \rightarrow \infty} \frac{\vv_p^H \vv_q}{\sqrt{N_v N_h N_f}} = 0,
\end{equation}
under Assumption \ref{assumption1} except for the special case of
\begin{align}
&\exists p, q \quad s.t. \quad {\phi _p + \phi _q} = \pm \pi, \theta _p = \theta _q, \tau_p = \tau_q. \label{Eq:phi_pi}
\end{align}
\end{lemma}
\begin{proof}
\quad \emph{Proof:} The proof can be found in Appendix \ref{proof:mutualOthogonal}.
\end{proof}
Remarks: Lemma \ref{lemma:mutualOthogonal} indicates that any two generalized steering vectors with non-identical angle or delay tend to be orthogonal to each other, with the only exception being Eq. (\ref{Eq:phi_pi}). In fact, the exception occurs because the steering vectors are identical in case of Eq. (\ref{Eq:phi_pi}):
\begin{align}
& {{\bf{a}}({\theta _{p,{\rm{ZOD}}}},{\phi _{p,{\rm{AOD}}}})} = {{\bf{a}}({\theta _{q,{\rm{ZOD}}}},{\phi _{q,{\rm{AOD}}}})}, \\
& \text {when } {\phi _{p,{\rm{ZOD}}} + \phi _{q,{\rm{ZOD}}}} = \pm \pi, {\theta _{p,{\rm{ZOD}}}} = {\theta _{q,{\rm{ZOD}}}}.
\end{align}
Note that such a special case is highly unlikely to happen, as a path departing from the back side of the antenna panel, e.g., with angle $\phi _p$ is very weak and has little probability to have exactly the same delay as a path departing from the front side with angle ${\phi}_q = \pi - \phi _p$ or ${\phi}_q = -\pi - \phi _p$. As a result, this special case is not to be concerned.
Lemma \ref{lemma:mutualOthogonal} will be applied in the proof of the subsequent Lemma \ref{lemma:orthLinearSpace}.

For ease of exposition, we ignore the spacial cases of Eq. (\ref{Eq:phi_pi}) by letting the range of $\phi$ contained within $[- \pi/2, \pi/2]$.
For a certain path $p$, we define a linear spaces $\mathcal{B}_p$:
\begin{align}
\mathcal{B}_p = \text{span}\{\vs_n: n \in \mathcal{M}_p  \},
\end{align}
where $\vs_n$ is the $n$-th column $(n = 1, \cdots, N_t N_f)$ of the matrix $\mS$ as in Eq. (\ref{Eq:Smat}). The set $\mathcal{M}_p$ is
\begin{equation}\label{Eq:M_p}
\mathcal{M}_p = \{m_1, m_2, \cdots, m_{S_p}\},
\end{equation}
such that
\begin{align}
\lim_{N_v, N_h, N_f \rightarrow \infty} |\frac{\vs_n^H \vv_p}{\sqrt{N_v N_h N_f}}| &> 0, \forall n \in \mathcal{M}_p, \\
\lim_{N_v, N_h, N_f \rightarrow \infty} |\frac{\vs_n^H \vv_p}{\sqrt{N_v N_h N_f}}| &= 0, \forall n \not\in \mathcal{M}_p,
\end{align}
with $\vv_p$ being the generalized steering vector as defined in Eq. (\ref{Eq:vp}). The linear space $\mathcal{B}_p$ can be regarded as the minimal space where the vector $\vv_p$ lives in. The dimensionality of $\mathcal{B}_p$ is $S_p$. In the same way we define a linear space $\mathcal{B}_q = \text{span}\{\vs_n: n \in \mathcal{M}_q  \}$ for a certain path $q$. Then we have

\begin{lemma}\label{lemma:orthLinearSpace}
For any $(\theta _{p}, {\phi _{p}}, \tau_p) \neq (\theta _{q}, {\phi _q}, \tau_q)$,  the two linear spaces $\mathcal{B}_p$ and $\mathcal{B}_q$ are asymptotically orthogonal when $N_v, N_h, N_f$ are large:
\begin{equation}
\mathcal{B}_p \bot \mathcal{B}_q \text{ as } {N_v, N_h, N_f \rightarrow \infty},
\end{equation}
or equivalently,
\begin{equation}
\mathcal{M}_p \cap \mathcal{M}_q = \emptyset \text{ as } {N_v, N_h, N_f \rightarrow \infty}.
\end{equation}
\end{lemma}
\begin{proof}
\quad \emph{Proof:} The proof can be found in Appendix \ref{proof:orthLinearSpace}.
\end{proof}
Remarks: In fact, when $N_v, N_h, N_f$ go to infinity, the generalized steering vector of a certain path lies in a column space of a submatrix of $\mS$. This submatrix is composed of $S_i$ columns of $\mS$, for $i = p, q$. Lemma \ref{lemma:orthLinearSpace} shows that when path $p$ and path $q$ are distinguishable in terms of either angle or delay, then they live in two orthogonal column spaces. In other words, the two paths will not interplay with each other after the orthogonal transformation by $\mS$. This effect will enable us to isolate the paths in mutually orthogonal column spaces of $\mS$, and thus make the signal processing and prediction easier.

\begin{lemma}\label{lemma:pronyOrder1}
Consider a uniformly sampled complex exponential signal $y(n) = \beta e^{j2\pi f n}, n \in \mathbb{N}^+ $, with fixed amplitude $\beta$ and frequency $f$. For any positive integer $N_d$, if two neighboring samples $y(m - 1)$ and $y(m)$ are known, then the Prony-based prediction at $N_d$ sample later, i.e., $\hat{y}(m + N_d)$, is error-free:
\begin{equation}
\hat{y}(m + N_d) = {y}(m + N_d).
\end{equation}
\end{lemma}
\begin{proof}
\quad \emph{Proof:} The proof can be found in Appendix \ref{proof:pronyOrder1}.
\end{proof}
Remarks: Lemma \ref{lemma:pronyOrder1} indicates that even with only two noiseless samples, Prony's method is able to predict any complex exponential signal with only one pole frequency at an arbitrary number of sample period later without prediction error.

Based on the lemmas, our theoretical result on the asymptotic performance of the PAD algorithm is shown in Theorem \ref{theoNoiseFreeEstimationErr}. We denote here the vectorized channel sample at time $t$ by $\tilde{\bm{\hbar}}_u(t)$, which is the noisy observation of  ${\bm{\hbar}}_u(t)$ in Eq. (\ref{Eq:hbaru}).
\begin{theorem}\label{theoNoiseFreeEstimationErr}
For an arbitrary delay $N_d \in \mathbb{N}^+$ and any UE velocity level, the asymptotic performance of the PAD algorithm yields:
\begin{equation}\label{Eq:noiseFreeEstimationErr}
\lim_{N_v, N_h, N_f \rightarrow \infty} \frac{\left\|{\hat{\bm{\hbar}}_u(t_{L+N_d}) - \bm{\hbar}}_u(t_{L+N_d})\right\|_2^2}{\left\|{\bm{\hbar}}_u(t_{L+N_d})\right\|_2^2} = 0,
\end{equation}
under Assumption \ref{assumption1} and the condition that two most recent channel samples are accurate enough, i.e.,
\begin{equation}\label{Eq:accurateSamples}
\lim_{N_v, N_h, N_f \rightarrow \infty} \frac{\left\|{\tilde{\bm{\hbar}}_u(t_{k}) - \bm{\hbar}}_u(t_{k})\right\|_2^2}{\left\|{\bm{\hbar}}_u(t_{k})\right\|_2^2} = 0, \forall k = L-1, L.
\end{equation}
\end{theorem}
\begin{proof}
\quad \emph{Proof:} The proof can be found in Appendix \ref{proof:theoNoiseFreeEstimationErr}.
\end{proof}
Remarks: Note that in order to achieve condition Eq. (\ref{Eq:accurateSamples}), we may need some non-linear signal processing techniques. See \cite{yin:16} as an example of how this condition can be fulfilled for a multi-cell massive MIMO scenario in the presence of pilot contamination.
The mild technical assumption, i.e., Assumption \ref{assumption1}, in Theorem \ref{theoNoiseFreeEstimationErr} is in general valid, since a finite number of multipath rays exist at a certain time and one ray is unlikely to have exactly the same angle and delay as another ray. Although in rich scattering environment as defined in CDL-A model of \cite{3gpp:38.901}, the number of paths $P$ can be as large as several hundreds, each path still has a unique tuple of $(\theta _{{\rm{ZOD}}}, {\phi _{{\rm{AOD}}}}, \tau)$.
In the special case when one tuple $(\theta _{{\rm{ZOD}}}, {\phi _{{\rm{AOD}}}}, \tau)$ is shared by more than one paths, i.e., $\exists p \neq q$, such that
\begin{equation}\label{Eq:sharedTuple}
 (\theta _{p, {\rm{ZOD}}}, {\phi _{p, {\rm{AOD}}}}, \tau_p) = (\theta _{q, {\rm{ZOD}}}, {\phi _{q, {\rm{AOD}}}}, \tau_q),
\end{equation}
the asymptotic performance of Theorem \ref{theoNoiseFreeEstimationErr} can be generalized.
We again ignore the special cases of Eq. (\ref{Eq:phi_pi}) for notational simplicity. The asymptotic performance under condition Eq. (\ref{Eq:sharedTuple}) is shown in Corollary \ref{coro:NcPathCommonTuple}.
\begin{corollary}\label{coro:NcPathCommonTuple}
Among all $P$ paths, if at most $N_c$ paths share exactly the same tuple $(\theta_{\rm{ZOD}}, {\phi}_{\rm{AOD}}, \tau)$, then for an arbitrary delay $N_d \in \mathbb{N}^+$ and any UE velocity level, the performance of the PAD algorithm satisfies:
\begin{equation}
\lim_{N_v, N_h, N_f \rightarrow \infty} \frac{\left\|{\hat{\bm{\hbar}}_u(t_{L+N_d}) - \bm{\hbar}}_u(t_{L+N_d})\right\|_2^2}{\left\|{\bm{\hbar}}_u(t_{L+N_d})\right\|_2^2} = 0,
\end{equation}
given that at least $2 N_c$ accurate enough samples are available.
\end{corollary}
\begin{proof}
\quad \emph{Proof:} The proof entails a generalization of Lemma \ref{lemma:pronyOrder1} to the case that the uniformly sampled signal is a sum of $N_c$ complex exponentials with time-invariant frequencies. It is straightforward to prove that in this case the prediction at $N_d \in \mathbb{N}^+$ sample later is also error-free as long as $2N_c$ accurate enough samples are available.
\end{proof}

In addition, for a narrowband system, e.g., $N_f = 1$, or $N_f$ is small, the frequency resolution is not sufficiently high. In such cases, the asymptotic result of our PAD algorithm is shown in Corollary \ref{coro:NcPathCommonTupleNB}.
\begin{corollary}\label{coro:NcPathCommonTupleNB}
Among all $P$ paths, if at most $N_c$ paths share exactly the same tuple $(\theta_{\rm{ZOD}}, {\phi}_{\rm{AOD}})$, then for an arbitrary delay $N_d \in \mathbb{N}^+$ and any UE velocity level, the performance of the PAD algorithm satisfies:
\begin{equation}
\lim_{N_v, N_h \rightarrow \infty} \frac{\left\|{\hat{\bm{\hbar}}_u(t_{L+N_d}) - \bm{\hbar}}_u(t_{L+N_d})\right\|_2^2}{\left\|{\bm{\hbar}}_u(t_{L+N_d})\right\|_2^2} = 0,
\end{equation}
given that at least $2N_c$ accurate enough samples are available.
\end{corollary}
\begin{proof}
\quad \emph{Proof:} This Corollary is a readily generalization of Corollary \ref{coro:NcPathCommonTuple}. Since $N_c$ paths share exactly the same angle, they are non-separable in spatial domain. Due to the limited bandwidth, the frequency resolution is also finite. As a result, the projection of the channel onto the spatial-frequency orthogonal basis results in a set of non-negligible angle-delay positions where the $N_c$ paths have non-zero coefficients. Each of these angle-delay positions has $N_c$ exponential signals. As a result, we would need $2 N_c$ accurate enough channel samples to compute the complete Prony coefficients in order to predict future channel. The full proof is omitted.
\end{proof}

Corollary \ref{coro:NcPathCommonTupleNB} indicates that our PAD still achieves very good performance even when only small bandwidth is used. However in this case more channel samples may be needed in order to compensate for the low frequency resolution.

\section{Dealing with noisy channel samples}\label{sec:estimationError}
The channel estimate at the base station is always corrupted by noise, which is expected to undermine the performances of our previous methods. Thus we propose to deal with noise with a supplementary method, which relies on the subspace structure of the channel sample matrix and the second-order long-term statistics of the noisy channel samples. It consists of the following two ingredients
\subsubsection{Tufts-Kumaresan's method}
The main idea of the Tufts-Kumaresan's method \cite{Tufts:1982} is to apply singular value decomposition (SVD) to the sample matrix, i.e., Eq. (\ref{Eq:Y}) or Eq. (\ref{Eq:mathcalG}), of the linear prediction equations, and then remove the contributions of small singular values. Taking the estimate of ${{\bf{p}}(u,n)}$
for example, the SVD of $\bm{\mathcal G}(u,n)$ can be written as:
\begin{align}\label{Eq:Gun_SVD}
\bm{\mathcal G}(u,n) &= \mU(u,n) \mathbf{\Sigma}(u,n) \mV^H(u,n) \\
&\approx \mU_s(u,n) \mathbf{\Sigma}_s(u,n) \mV_s^H(u,n),
\end{align}
where $\mathbf{\Sigma}_s(u,n)$ only contains the significant singular values of $\bm{\mathcal G}(u,n)$. The Tufts-Kumaresan's estimate of the Prony coefficients are given by
\begin{equation}\label{Eq:phat_tk_un}
{\hat{\bf{p}}_\text{tk}(u,n)} = - \mV_s(u,n) \mathbf{\Sigma}_s^{-1}(u,n)\mU_s^H(u,n){{\bf{g}}(u,n)} .
\end{equation}
Note that $\mathbf{\Sigma}_s(u,n)$ can be obtained in a way that the minimum number of singular values satisfy
\begin{equation}
\operatorname{tr}\left\{ \mathbf{\Sigma}_s(u,n) \right\} \geq \gamma_\text{tk}\operatorname{tr}\left\{ \mathbf{\Sigma}(u,n) \right\},
\end{equation}
where the threshold $\gamma_\text{tk}$ is no greater than 1, i.e., $\gamma_\text{tk} = 0.99$.
\subsubsection{Channel denoising with statistical information}
The noisy channel samples between all base station antennas and the $u$-th UE antenna at time $t$ and frequency $f$
can be modeled as ${{\tilde{\bf{h}}_u}(f,t)} \in \mathbb{C}^{1 \times N_t}$:
\begin{equation}\label{Eq:hhatuft}
{{\tilde{\bf{h}}_u}(f,t)} = {{{\bf{h}}_u}(f,t)} + \vn_u (f,t),
\end{equation}
where ${{{\bf{h}}_u}(f,t)}$ is the accurate channel and $\vn_u (f,t)$ is the independent and identically distributed (i.i.d.) complex Gaussian noise with zero-mean and covariance $\sigma_n^2$. It is easy to obtain the covariance matrix of the noisy channel at the base station:
\begin{equation}\label{Eq:Rhatf}
\tilde \mR = \mathbb{E} \left\{{{\tilde{\bf{H}}^H}(f,t)} {{\tilde{\bf{H}}}(f,t)} \right\},
\end{equation}
where the expectation is taken over time, frequency, or both. ${{\tilde{\bf{H}}}(f,t)}$ is defined as
\begin{equation}
{{\tilde{\bf{H}}}(f,t)} \triangleq {[\begin{array}{*{20}{c}}
{{\tilde{\bf{h}}_1^T}(f,t)}&{{\tilde{\bf{h}}^T_2}(f,t)}& \cdots &{{\tilde{\bf{h}}^T_{{N_r}}}(f,t)}
\end{array}]^T}.
\end{equation}

From Eq. (\ref{Eq:hhatuft}) we have
\begin{equation}\label{Eq:Rf}
\tilde\mR = \mR + N_r \sigma_n^2 \mI,
\end{equation}
where
\begin{equation}
\mR = \mathbb{E} \left\{{{{\bf{H}}^H}(f,t)} {{{\bf{H}}}(f,t)} \right\},
\end{equation}
with ${{{\bf{H}}}(f,t)}$ being the accurate counterpart of ${{\tilde{\bf{H}}}(f,t)}$.
Due to the large number of base station antennas and the limited scattering environment, the channel covariance matrix $\mR$ has a low-rankness property \cite{yin:13} \cite{adhikary:13}, which means a fraction of the eigenvalues of $\mR$ are very close to zero. Thus we may exploit this property to have an estimate of the power of noise. The eigen-decomposition of $\tilde \mR$ is written as $\tilde \mR = \tilde{\mU} \tilde{\mathbf{\Sigma}} \tilde{\mU}^H$ where $\tilde{\mathbf{\Sigma}} = {\mathop{\rm diag}\nolimits} \{ {\sigma_1,...,\sigma_{N_t}}\}$ and $\tilde{\mU}$ contains all the eigen-vectors of $\tilde \mR$. The estimate of the noise power ${\sigma}^2_n$ is obtained by simply averaging the smallest eigenvalues of $\tilde{\mR}$. A linear filter $\bm{\mathcal{W}}$ can be derived for channel denoising purpose:
\begin{equation}\label{Eq:Wf}
\bm{\mathcal{W}} = \mathop {\arg \min }\limits_{\bm{\mathcal{W}}} \mathbb{E} \left\{\| {{\tilde{\bf{H}}}(f,t)}\bm{\mathcal{W}} - {{{\bf{H}}}(f,t)} \|_F^2 \right\}.
\end{equation}
The solution is given by Proposition \ref{propBoundedness}.
\begin{proposition}\label{propBoundedness}
The linear solution to the optimization problem of Eq. (\ref{Eq:Wf}) yields
\begin{equation}\label{Eq:Wfsolution}
\bm{\mathcal{W}} = \tilde{\mU} \mD \tilde{\mU}^H,
\end{equation}
where $\mD$ is a diagonal matrix with its $i$-th ($i = 1, \cdots, N_t$) diagonal entry being ${\frac{{{\sigma _i} - N_r \hat \sigma _n^2}}{{{\sigma _i}}}}$.

\end{proposition}
\begin{proof}
\quad \emph{Proof:} The derivation is based on the linear minimum mean square error (LMMSE) criterion. One may readily obtain this result by computing the partial derivative of the trace of error covariance matrix with respect to ${{{\bf{H}}}(f,t)}$ \cite{Hjorungnes:07} and letting it equal to zero. The detailed proof is omitted.
\end{proof}

Note that the covariance matrix is computed based on samples of all $N_r$ UE antennas, since the scattering environments experienced by all co-located $N_r$ antennas is very similar. The denoising filter $\bm{\mathcal{W}}$ can also be built for each UE antenna, however with higher complexity.

\section{Numerical Results}\label{sec:numericalResult}
This section contains simulation results of our proposed channel prediction schemes. The basic simulation parameters are listed in Table \ref{Tb:basicParas}. We mainly adopt the CDL-A channel model, unless otherwise notified. The number of multipath is 460, i.e., for each UE, there are 23 clusters of multipath with each cluster containing 20 rays. The Root Mean Square (RMS) angular spreads of the azimuth departure angle, elevation departure angle, azimuth arrival angle, and elevation arrival angle are $87.1^\circ$, $33.6^\circ$, $102.1^\circ$, and $24.7^\circ$ respectively.
Consider a typical setting of 5G at 3.5 GHz with 30 kHz of subcarrier spacing. Each slot, with a duration of 0.5 ms, contains 14 OFDM symbols. We assume that UE sends one SRS signal in each slot, which means one channel sample is available every $\triangle T = 0.5$ ms.
\begin{table}[h]
\caption{Basic simulation parameters}\label{Tb:basicParas}
\centering
\begin{tabular}{|p{2.1cm}|p{5.5cm}|}
  \hline
  Scenario          & 3D Urban Macro (3D UMa) \\
  \hline
  Carrier frequency           & 3.5 GHz \\
    \hline
  Subcarrier spacing         & 30 kHz \\
    \hline
  Bandwidth  &  20 MHz (51 RBs)\\
    \hline
  Number of UEs  &  8\\
    \hline
  BS antenna configuration    & $(\underline{M},\underline{N},\underline{P},\underline{M}_g,\underline{N}_g) = (2,8,2,1,1)/(4,8,2,1,1)$, $(dH, dV) = (0.5, 0.8)\lambda$, the polarization angles are $\pm 45^\circ$ \\
    \hline
  UE antenna configuration    & $(\underline{M},\underline{N},\underline{P},\underline{M}_g,\underline{N}_g) = (1,1,2,1,1)$,  the polarization angles are $0^\circ$ and $90^\circ$\\
    \hline
  Channel model     & CDL-A \\
    \hline
  Delay spread       & 300 ns \\
    \hline
  DL precoder       & EZF \\
    \hline
  UE receiver       & MMSE-IRC \\
    \hline
  CSI delay       & 4 ms \\
    \hline
  Number of paths       & 460 \\
    \hline
\end{tabular}
\end{table}
The tuple $(\underline{M},\underline{N},\underline{P},\underline{M}_g,\underline{N}_g)$ in Table \ref{Tb:basicParas} means the antenna array is composed of $\underline{M}_g \underline{N}_g$ panels of UPAs with $\underline{M}_g$ being the number of panels in a column and $\underline{N}_g$ the number of panels in a row. Furthermore, each antenna panel has $\underline{M}$ rows and $\underline{N}$ columns of antenna elements in the same polarization. The number of polarizations is always $\underline{P} = 2$. Thus the total number of antennas is $\underline{M} \times \underline{N} \times \underline{P} \times \underline{M}_g \times \underline{N}_g$ for a certain BS or UE. We consider 20 MHz of bandwidth where one channel estimate per each resource block (RB) is available in frequency domain.
The DL precoder is the Eigen Zero-Forcing (EZF) \cite{Sun:2010TSP} and the receiver at UE side is Minimum Mean Square Error - Interference Rejection Combining (MMSE-IRC). Each UE receives one data stream from BS. The DL spectral efficiency is our main performance metric, which is computed as $\log_2(1+\text{SINR})$ averaged over the whole bandwidth.

We first ignore the channel sample error and plot the spectral efficiency as a function of SNR at UE side. We show the performances of our proposed vector Prony algorithm (denoted by Vec Prony) and PAD algorithm with 60 km/h of velocity level for all UEs in Fig. \ref{fig:SE_32T2R_60km} for the case $N_t = 32$ and in Fig. \ref{fig:SE_64T2R_60km} for the case $N_t = 64$. The performances of 0, 3, and 60 km/h of UE speeds without channel prediction are also added as reference curves.
The curves labeled by ``FIR Wiener" are obtained by the use of a classical linear predictor based on AR modeling of channel variations (for instance as proposed by \cite{Truong:2013}).
In all figures of this section, $N$ denotes the order of the predictor.
We may observe from Fig. \ref{fig:SE_32T2R_60km} and Fig. \ref{fig:SE_64T2R_60km} that our proposed algorithms nearly approach the ideal case where UEs are stationary and the channels are time-invariant. It is interesting to note that the vector Prony method and the PAD method both outperform the low-mobility scenario of 3 km/h without channel prediction.
Note that the FIR Wiener predictor gives only moderate prediction gains. In fact it is not performing as well because models that account for multipath space-time structure (such as the one in \cite{3gpp:38.901}) do not necessarily conform with the simple AR(1) channel aging model.

\begin{figure}[h]
  \centering
  \includegraphics[width=3.5in]{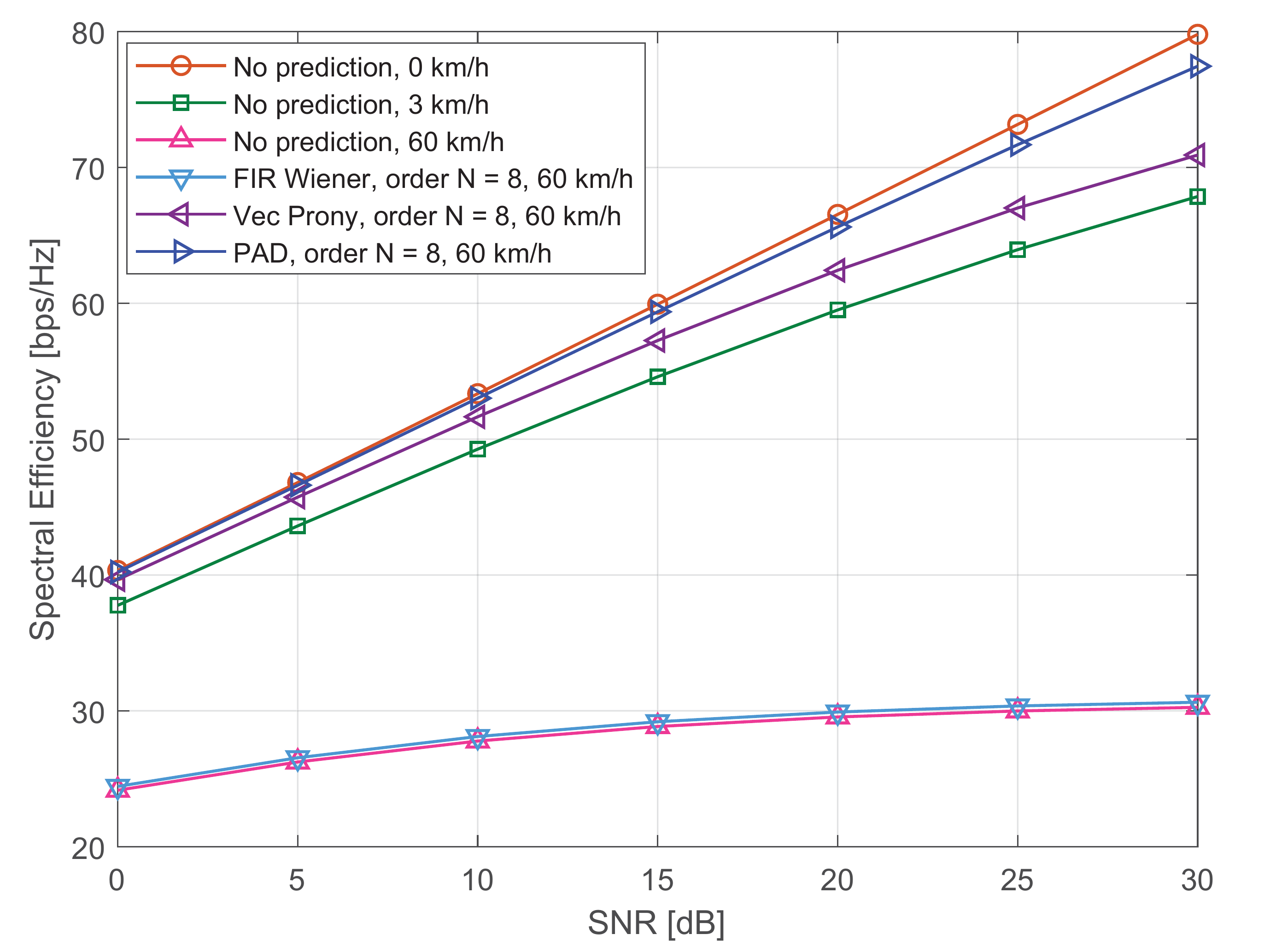}\\
  \caption{The spectral efficiency vs. SNR, $N_t = 32$, noise-free channel samples, CDL-A model, 4 ms of CSI delay.} \label{fig:SE_32T2R_60km}
\end{figure}

\begin{figure}[h]
  \centering
  \includegraphics[width=3.5in]{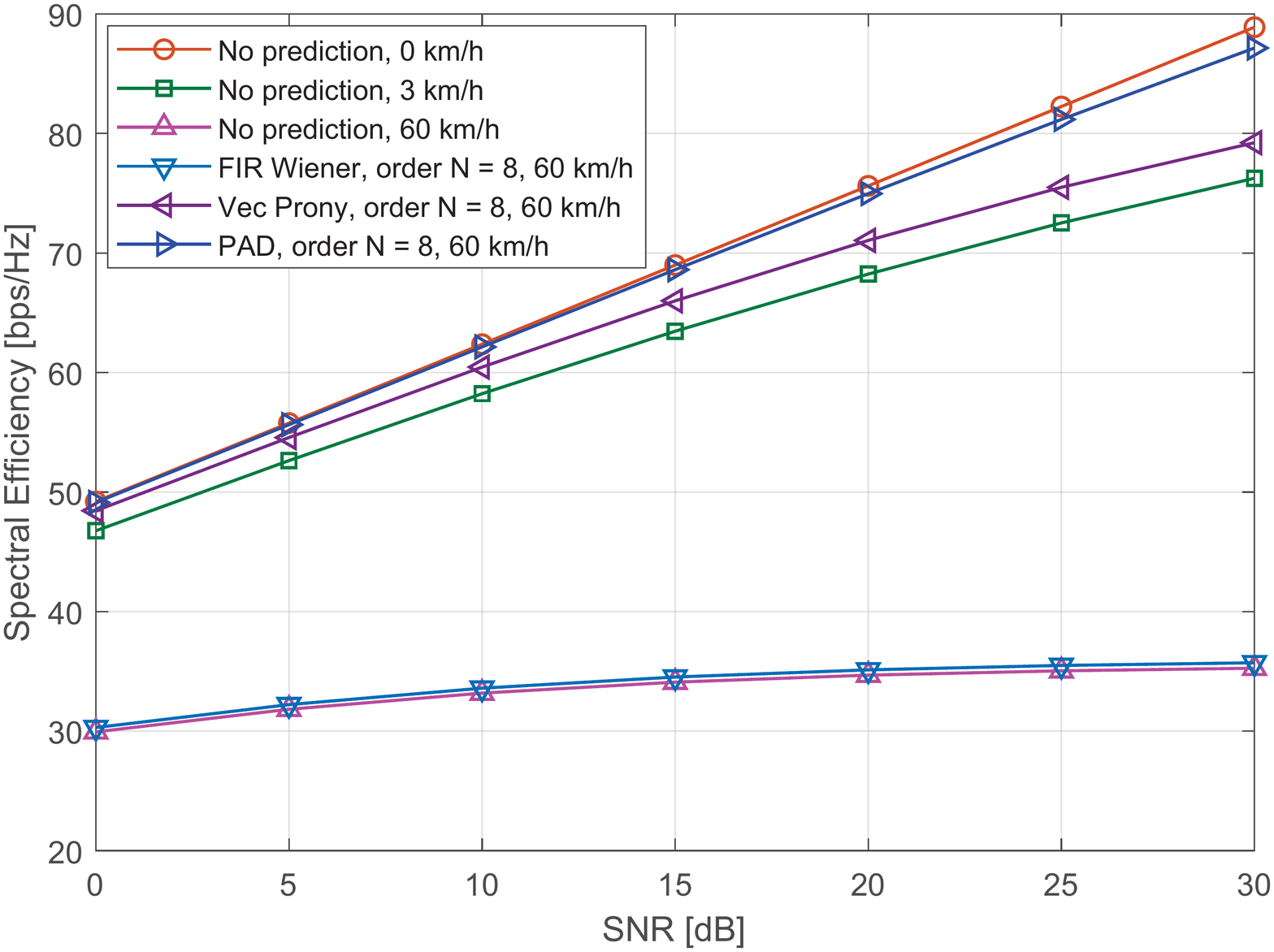}\\
  \caption{The spectral efficiency vs. SNR, $N_t = 64$, noise-free channel samples, CDL-A model, 4 ms of CSI delay.} \label{fig:SE_64T2R_60km}
\end{figure}

\begin{figure}[h]
  \centering
  \includegraphics[width=3.5in]{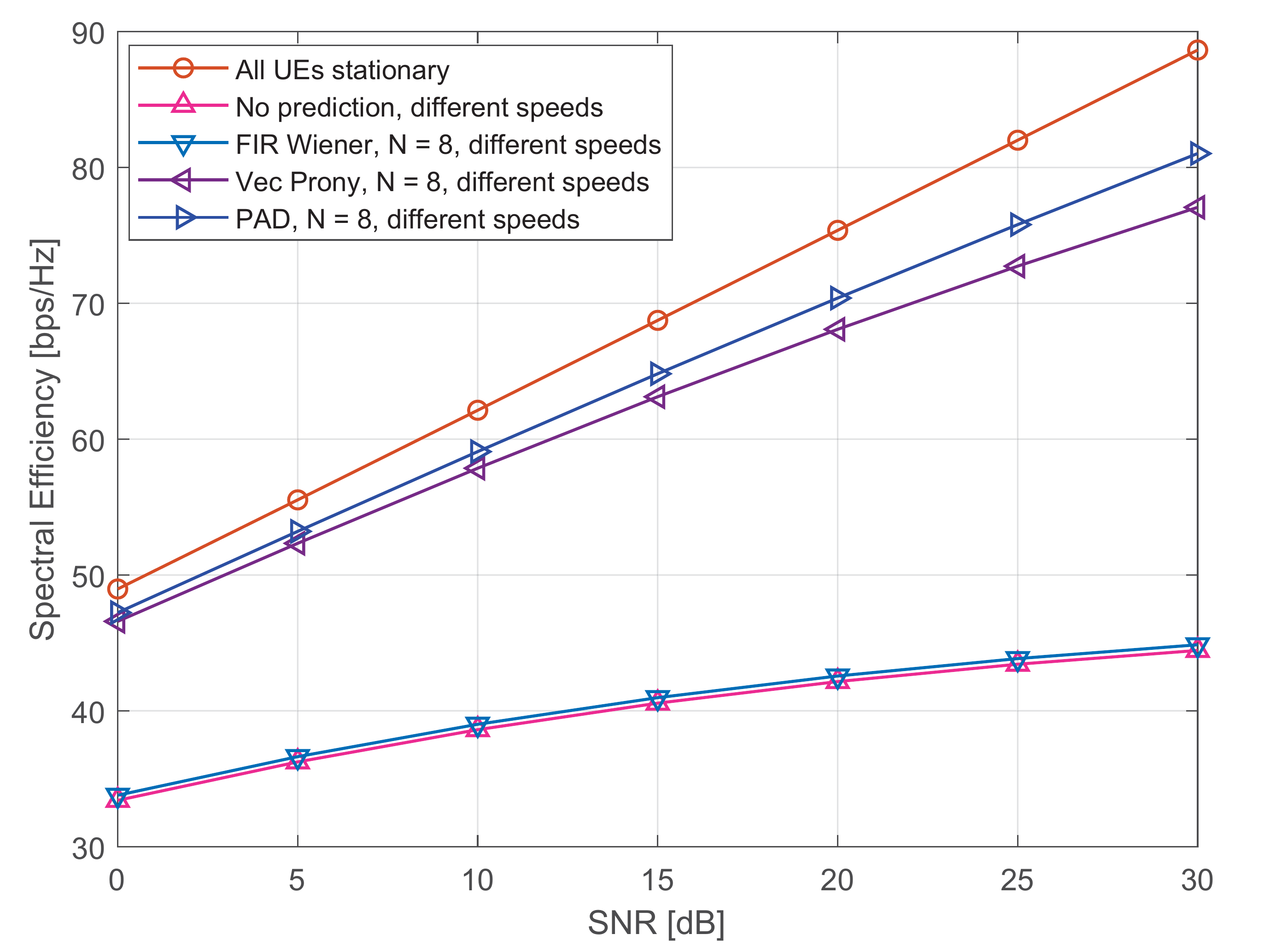}\\
  \caption{The spectral efficiency vs. SNR, $N_t = 64$, noise-free channel samples, CDL-A model, 4 ms of CSI delay. Multiple UEs have different speeds: two UEs move at 3 km/h, two at 30 km/h, two at 60 km/h, and two at 90 km/h.} \label{fig:SE_64T2R_differentSpeeds}
\end{figure}

Fig. \ref{fig:SE_64T2R_differentSpeeds} shows the performances of a realistic setting where multiple UEs move at different speeds. For simplicity, we keep the order of all predictors to $N=8$ despite the fact that low-mobility UEs may only need a smaller order. One can observe that our proposed methods works well in this setting.

Fig. \ref{fig:PredErr_Ant} shows the channel prediction error as a function of BS antennas. The channel prediction error is defined as
\begin{align}
\varepsilon = 10 \log\left\{\mathbb{E} \frac{\left\|\hat{\mH} - \mH \right\|^2_F}{\left\|\mH \right\|^2_F} \right\},
\end{align}
where $ \mH \in \mathbb{C}^{N_r \times N_t}$ and $\hat{\mH}\in \mathbb{C}^{N_r \times N_t}$ are the channel matrix and its prediction respectively. The expectation is taken over time, frequency, and UEs. In Fig. \ref{fig:PredErr_Ant} the numbers of BS antennas are $N_t = 4, 8, 32, 128, 512, 2048$, with corresponding layouts being $(\underline{M},\underline{N},\underline{P},\underline{M}_g,\underline{N}_g) = (1,2,2,1,1)$, $(1,4,2,1,1)$, $(2,8,2,1,1)$, $(4,16,2,1,1)$, $(8,32,2,1,1)$, $(16, 64, 2, 1, 1)$ respectively. We may observe from Fig. \ref{fig:PredErr_Ant} that the prediction accuracy of our PAD algorithm keeps increasing with the number of BS antennas, which is inline with Theorem \ref{theoNoiseFreeEstimationErr}.

\begin{figure}[h]
  \centering
  \includegraphics[width=3.5in]{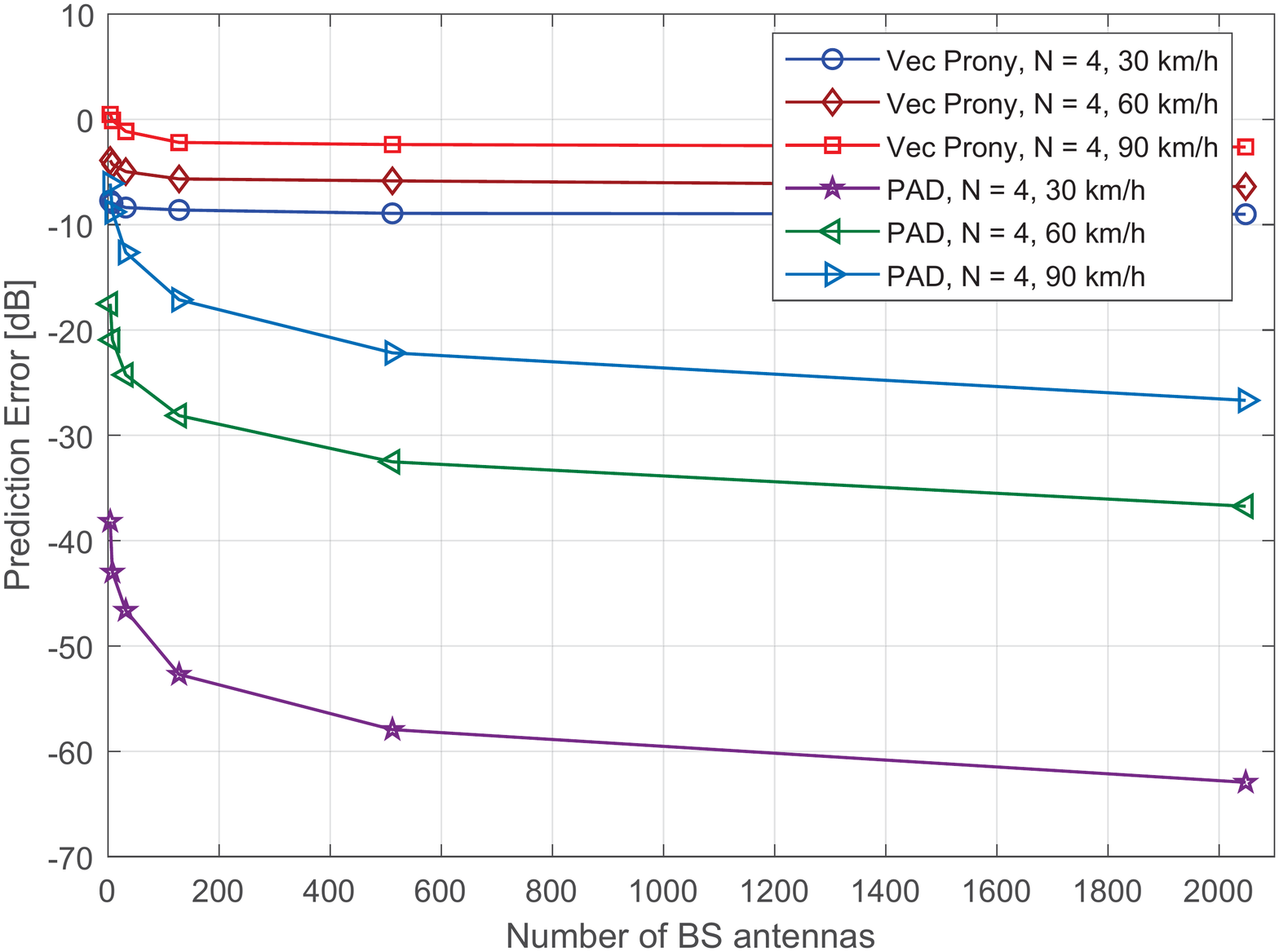}\\
  \caption{The prediction error vs. the number of BS antennas $N_t$, noise-free channel samples, CDL-A model, 4 ms of CSI delay.} \label{fig:PredErr_Ant}
\end{figure}

In Fig. \ref{fig:SE_64T2R_CDL_D} we change the channel model to CDL-D \cite{3gpp:38.901}, where a LOS component is present among all paths. We show the performances of the algorithms with 90 km/h of moving speed for all UEs. The gains are still significant even in such high mobility scenarios. Note that in this figure we have only 2 rows and 8 columns of dual-polarized BS antennas and therefore the spatial resolution is quite limited. More gains will be expected with larger number of antennas.
\begin{figure}[h]
  \centering
  \includegraphics[width=3.5in]{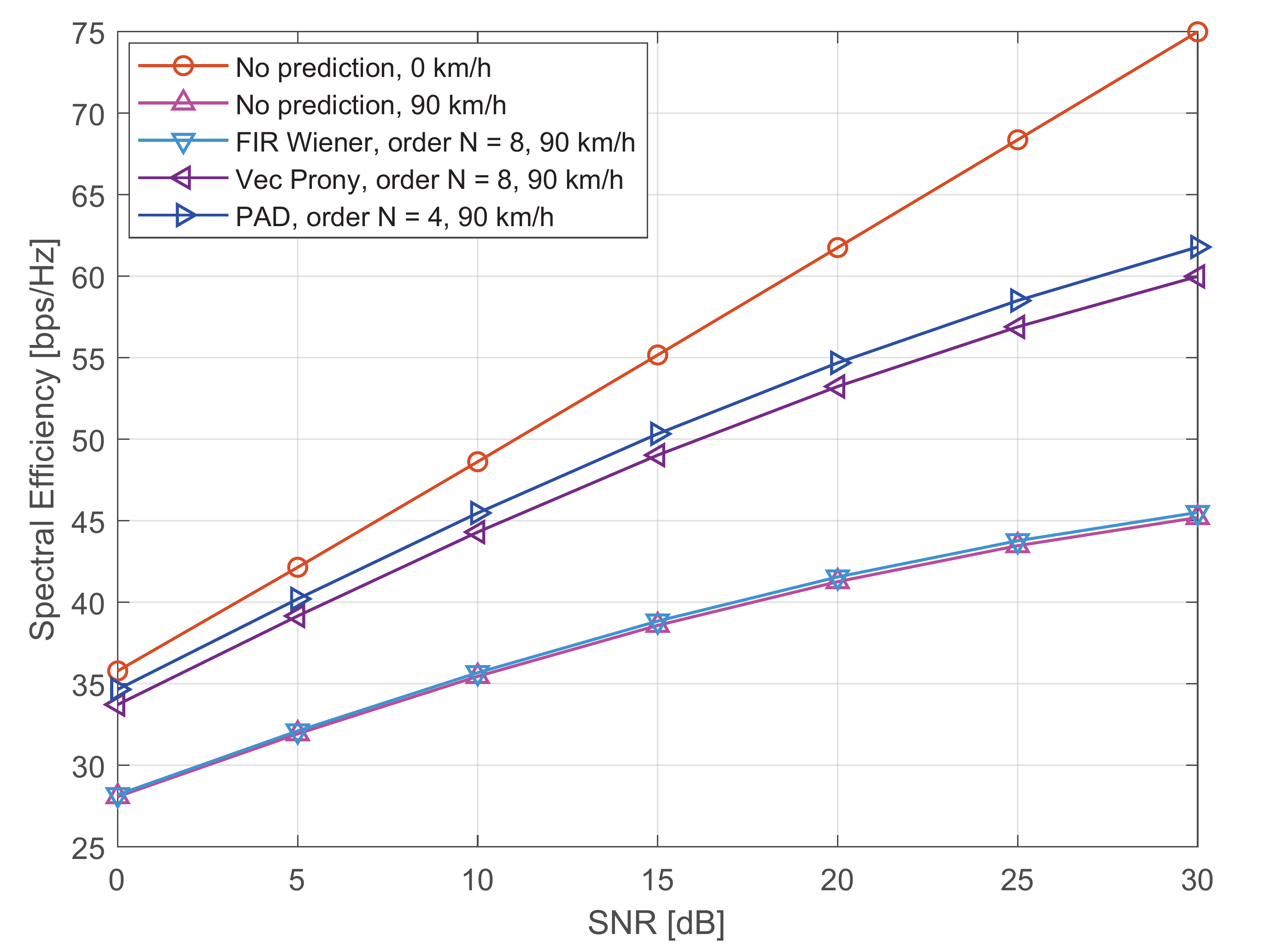}\\
  \caption{The spectral efficiency vs. SNR, $N_t = 32$, noise-free channel samples, CDL-D model with 261 paths, 4 ms of CSI delay. The RMS angular spreads of the azimuth departure angle, elevation departure angle, azimuth arrival angle, and elevation arrival angle are $47.9^\circ$, $7.1^\circ$, $89.9^\circ$, and $5.4^\circ$ respectively.} \label{fig:SE_64T2R_CDL_D}
\end{figure}

Now the channel estimation error is taken into consideration, assuming the ratio between the channel power and the power of estimation noise is 20 dB. we plot the performances of the vector Prony-based algorithm and the PAD algorithm, both combined with the denoising methods given by Sec. \ref{sec:estimationError} in Fig. \ref{fig:SE_noisyCHE_30km}.
\begin{figure}[h]
  \centering
  \includegraphics[width=3.5in]{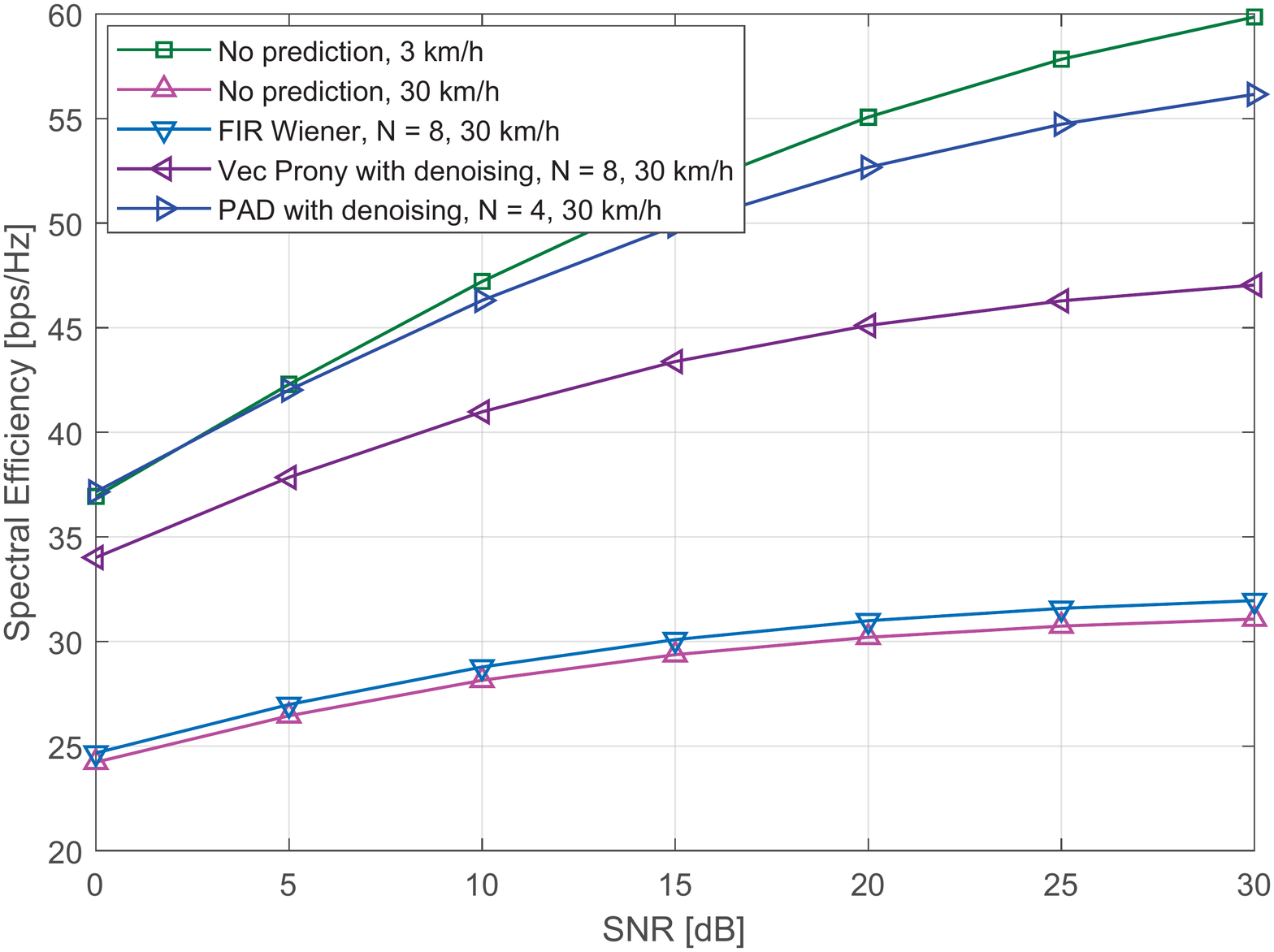}\\
  \caption{The spectral efficiency vs. SNR, $N_t = 32$, noisy channel samples, CDL-A model, 4 ms of CSI delay. } \label{fig:SE_noisyCHE_30km}
\end{figure}
As we may observe, our proposed PAD algorithm combined with denoising methods in moderate mobility scenario of 30 km/h is very close to the low-mobility scenario of 3 km/h, and thus proves its robustness when channel samples are corrupted by noise.

\section{Conclusions}\label{conclusion}
In this paper we addressed the practical challenge of massive MIMO -- the mobility problem. We proposed vector Prony algorithm and Prony-based angular-delay domain channel prediction method which are based on the angle-delay-Doppler structure of the channel.
Our theoretical analysis proves that the proposed PAD method is able to achieve asymptotically error-free prediction, provided that only two accurate channel samples are available. In case the channel samples are inaccurate, we proposed to combine our algorithms with denoising methods based on the subspace structure and the long-term statistics of the channel observations.
Simulation results show that in a practical setting with moderate to high mobility, several milliseconds of CSI delay, and rich scattering environment, our proposed methods achieve nearly ideal performance of stationary setting even with moderate number of base station antennas and relatively small bandwidth.

Finally, our work also opens a new prospect to further enhance the spectral efficiency of massive MIMO by offering more multiplexing gains.
In practice, the maximum number of simultaneously served UEs is primarily determined by coherence time and coherence bandwidth \cite{marzetta:10a}. Although demonstrated in time-domain prediction, our methods can also be generalized to frequency domain extrapolation. As a result, they have the potential of greatly reduce the time-frequency resources consumed by pilots of one user and thus lead to higher multi-user multiplexing gains given a fixed coherence time and coherence bandwidth.

\appendix
\subsection{Proof of Lemma \ref{lemma:mutualOthogonal}:}\label{proof:mutualOthogonal}
\begin{proof}
We decompose the proof into three sub-problems below. Sub-problem 1:
\begin{align}\label{Eq:a_v_orth}
\lim_{N_v \rightarrow \infty} \frac{\va_v(\theta_p)^H \va_v(\theta_q)}{\sqrt{N_v}} = 0, \text{ when } \theta_p \neq \theta_q.
\end{align}
Sub-problem 2:
\begin{align}\label{Eq:a_h_orth}
& \lim_{N_h \rightarrow \infty} \frac{{{\bf{a}}_h}(\theta_p ,\phi_p )^H {{\bf{a}}_h}(\theta_q ,\phi_q )}{\sqrt{N_h}} = 0, \\
& \text{ when } {\sin ({\theta _p})\sin ({\phi _p})} \neq \sin ({\theta _q})\sin ({\phi _q}). \nonumber
\end{align}
And sub-problem 3:
\begin{align}\label{Eq:b_tau_orth}
\lim_{N_f \rightarrow \infty} \frac{\vb(\tau_p)^H \vb(\tau_q)}{\sqrt{N_f}} = 0, \text{ when } \tau_p \neq \tau_q.
\end{align}

Starting from sub-problem 1, we may write
\begin{small}
\begin{align}
&{{\bf{a}}_v}{({\theta _p})^H}{{\bf{a}}_v}({\theta _q}) = \sum\limits_{n = 0}^{{N_v} - 1} {{e^{ - j2\pi \frac{{n{D_v}\cos ({\theta _p})}}{{{\lambda _0}}}}}{e^{j2\pi \frac{{n{D_v}\cos ({\theta _q})}}{{{\lambda _0}}}}}} \\
& = \sum\limits_{n = 0}^{{N_v} - 1} {{e^{j2\pi \frac{{n{D_v}\left( {\cos ({\theta _q}) - \cos ({\theta _p})} \right)}}{{{\lambda _0}}}}}}  = \frac{{1 - {e^{j2\pi \frac{{{N_v}{D_v}\left( {\cos ({\theta _q}) - \cos ({\theta _p})} \right)}}{{{\lambda _0}}}}}}}{{1 - {e^{j2\pi \frac{{{D_v}\left( {\cos ({\theta _q}) - \cos ({\theta _p})} \right)}}{{{\lambda _0}}}}}}} \nonumber
\end{align}
\end{small}
Since $\theta _p, \theta _q \in [0, \pi]$ and $\theta _p \neq \theta _q$, we can easily see that $|{{\bf{a}}_v}{({\theta _p})^H ({\theta _q})}|$ is a finite value and that
\begin{equation}
\lim_{N_v \rightarrow \infty} \frac{\va_v(\theta_p)^H \va_v(\theta_q)}{\sqrt{N_v}} = 0.
\end{equation}
Then, for sub-problem 2, we have
\begin{align}
& {{\bf{a}}_h}{({\theta _p},{\phi _p})^H}{{\bf{a}}_h}({\theta _q},{\phi _q})\\
 = & \sum\limits_{n = 0}^{{N_h} - 1} {{e^{j2\pi \frac{{n{D_h}\left( {\sin ({\theta _q})\sin ({\phi _q}) - \sin ({\theta _p})\sin ({\phi _p})} \right)}}{{{\lambda _0}}}}}} \\
 = & \frac{{1 - {e^{j2\pi \frac{{{N_h}{D_h}\left( {\sin ({\theta _q})\sin ({\phi _q}) - \sin ({\theta _p})\sin ({\phi _p})} \right)}}{{{\lambda _0}}}}}}}{{1 - {e^{j2\pi \frac{{{D_h}\left( {\sin ({\theta _q})\sin ({\phi _q}) - \sin ({\theta _p})\sin ({\phi _p})} \right)}}{{{\lambda _0}}}}}}}.
\end{align}
When the term ${\sin ({\theta _q})\sin ({\phi _q}) - \sin ({\theta _p})\sin ({\phi _p})}$ is not zero, we may readily see that ${{\bf{a}}_h}{({\theta _p},{\phi _p})^H}{{\bf{a}}_h}({\theta _q},{\phi _q})$ is a finite value and sub-problem 2 is proved. Sub-problem 3 has a similar structure as sub-problem 1 and the proof is omitted. We may further write
\begin{small}
\begin{align}
& \frac{{{\bf{v}}_p^H{{\bf{v}}_q}}}{\sqrt{{N_v}{N_h}{N_f}}} =  \frac{{{{\left( {{\bf{b}}({\tau _p}) \otimes {\bf{a}}({\theta _p},{\phi _p})} \right)}^H}\left( {{\bf{b}}({\tau _q}) \otimes {\bf{a}}({\theta _q},{\phi _q})} \right)}}{\sqrt{{N_v}{N_h}{N_f}}}\\ \nonumber
& = \frac{{{{\left( {{\bf{b}}({\tau _p}) \otimes {{\bf{a}}_h}({\theta _p},{\phi _p}) \otimes {{\bf{a}}_v}({\theta _p})} \right)}^H}\left( {{\bf{b}}({\tau _q}) \otimes {{\bf{a}}_h}({\theta _q},{\phi _q}) \otimes {{\bf{a}}_v}({\theta _q})} \right)}}{\sqrt{{N_v}{N_h}{N_f}}}\\ \nonumber
& = \frac{{\left( {{\bf{b}}{{({\tau _p})}^H}{\bf{b}}({\tau _q})} \right)\left( {{{\bf{a}}_h}{{({\theta _p},{\phi _p})}^H}{{\bf{a}}_h}({\theta _q},{\phi _q})} \right)\left( {{{\bf{a}}_v}{{({\theta _p})}^H}{{\bf{a}}_v}({\theta _q})} \right)}}{\sqrt{{N_v}{N_h}{N_f}}}
\end{align}
\end{small}

It is clear that as long as one of the three terms ${\bf{b}}{{({\tau _p})}^H}{\bf{b}}({\tau _q})/\sqrt{N_f}$, ${{{\bf{a}}_h}{{({\theta _p},{\phi _p})}^H}{{\bf{a}}_h}({\theta _q},{\phi _q})}/\sqrt{N_h}$, and ${{\bf{a}}_v}{{({\theta _p})}^H}{{\bf{a}}_v}({\theta _q})/\sqrt{N_v}$ goes to zero, then Eq. (\ref{Eq:vpvq}) holds, since the absolute values of the three terms are no greater than 1. According to the proof of sub-problem 1, we conclude that condition $\theta_p \neq \theta_q$ is a sufficient condition of the equality Eq. (\ref{Eq:vpvq}), so is the condition $\tau_p \neq \tau_q$. We examine the case when $\theta_p = \theta_q$ and $\tau_p = \tau_q$. If $\theta_p = \theta_q = 0 \text{ or } \pi$, then $\phi_p = \phi_q = 0$ due to our definition in Eq. (\ref{Eq:theta0_pi}). In this case $(\theta _{p}, {\phi _{p}}, \tau_p) = (\theta _{q}, {\phi _q}, \tau_q)$, which contradicts the condition Eq. (\ref{Eq:noiseFreeCond}). While if $\theta_p = \theta_q \neq 0 \text{ or } \pi$, then according to sub-problem 2 the term ${{{\bf{a}}_h}{{({\theta _p},{\phi _p})}^H}{{\bf{a}}_h}({\theta _q},{\phi _q})}/N_h$ does not go to zero only when $\sin ({\theta _p}) = \sin ({\theta _q})$, or equivalently, $\theta _p + \theta _q = 0 \text{ or } \pi$. Thus, Lemma \ref{lemma:mutualOthogonal} is proved.
\end{proof}

\subsection{Proof of Lemma \ref{lemma:orthLinearSpace}:}\label{proof:orthLinearSpace}
\begin{proof}
Without loss of generality, we assume $0 < \theta_p < \theta_q < \pi, -\pi/2 < \phi_p < \phi_q < \pi/2, 0 < \tau_p < \tau_q$ and ignore the edge cases of $\theta = 0 \text{ or } \pi, \phi = \pm \pi/2$. Since $\theta_p \neq \theta_q$, we may define a positive value $\Delta_\theta$ such that
\begin{equation}
    0 \leq \theta_p - \Delta_\theta < \theta_p + \Delta_\theta < \theta_q - \Delta_\theta < \theta_q + \Delta_\theta \leq \pi.
\end{equation}
Similarly, define a positive value $\Delta_\phi$ and $\Delta_\tau$ such that
\begin{align*}
\begin{gathered}
- \frac{\pi}{2} \leq \phi_p - \Delta_\phi < \phi_p + \Delta_\phi < \phi_q - \Delta_\phi < \phi_q + \Delta_\phi \leq \frac{\pi}{2} \\
    0 \leq \tau_p - \Delta_\tau < \tau_p + \Delta_\tau < \tau_q - \Delta_\tau < \tau_q + \Delta_\tau.
\end{gathered}
\end{align*}
Define two joint probability density functions (PDF) as $p_p(\theta, \phi, \tau)$ and $p_q(\theta, \phi, \tau)$ such that
\begin{align}
& \int_{{\theta _p} - {\Delta _\theta }}^{{\theta _p} + {\Delta _\theta }} {\int_{{\phi _p} - {\Delta _\phi }}^{{\phi _p} + {\Delta _\phi }} {\int_{{\tau _p} - {\Delta _\tau }}^{{\tau _p} + {\Delta _\tau }} {{p_p}(\theta ,\phi ,\tau )} } } d\theta d\phi d\tau  = 1 \\
& \int_{{\theta _q} - {\Delta _\theta }}^{{\theta _q} + {\Delta _\theta }} {\int_{{\phi _q} - {\Delta _\phi }}^{{\phi _q} + {\Delta _\phi }} {\int_{{\tau _q} - {\Delta _\tau }}^{{\tau _q} + {\Delta _\tau }} {{p_q}(\theta ,\phi ,\tau )} } } d\theta d\phi d\tau  = 1.
\end{align}
Define two sets of tuples $\Omega_p$ and $\Omega_q$ such that for $i = p, q$,
\begin{align}
\Omega_i \defi & \left\{ (\theta ,\phi ,\tau )|\theta  \in ({\theta _i} - {\Delta _\theta },{\theta _i} + {\Delta _\theta }) \cap \right. \\
 & \left. \phi  \in ({\phi _i} - {\Delta _\phi },{\phi _i} + {\Delta _\phi }) \cap \tau  \in ({\tau _i} - {\Delta _\tau },{\tau _i} + {\Delta _\tau }) \right\}. \nonumber
\end{align}
The joint PDF satisfies
\begin{align}
0 < {p_i} < \infty , \text{ when } (\theta ,\phi ,\tau ) \in \Omega_i, i = p,q.
\end{align}
For ease of exposition, we let the PDF be a constant within its angular and delay support, i.e.,
\begin{align}
{p_i} =  \frac{1}{8 \Delta _\theta \Delta _\phi \Delta _\tau} , \text{ when } (\theta ,\phi ,\tau ) \in \Omega_i, i = p,q.
\end{align}
Assuming each path realization is generated according to the joint PDF $p_i$ and has a i.i.d. random phase, we can write the covariance matrix of the random paths as
\begin{equation}
\mathcal{R}_i = \int\limits_{(\theta ,\phi ,\tau ) \in {\Omega _i}} {{\bf{v}}(\theta ,\phi ,\tau ){\bf{v}}{{(\theta ,\phi ,\tau )}^H}{p_i}} d\theta d\phi d\tau,
\end{equation}
where $\vv(\theta, \phi, \tau) = {{\bf{b}}({\tau}) \otimes {{\bf{a}}_h}({\theta},{\phi}) \otimes {{\bf{a}}_v}({\theta})}$. 
We now prove that any normalized $\vv(\theta, \phi, \tau)$ with $(\theta, \phi, \tau) \not\in \Omega _i$ falls into the null space of $\mathcal{R}_i$. Consider an arbitary tuple $(\theta_o, \phi_o, \tau_o) \not\in \Omega _i$. The corresponding generalized vector with tuple $(\theta_o, \phi_o, \tau_o)$ is denoted as $\vv_o$.
Then we can write
\begin{align}
& \mathop {\lim }\limits_{{N_v},{N_h},{N_f} \to \infty } \frac{{{{\bf{v}}_o}^H}}{{\sqrt {{N_v}{N_h}{N_f}} }}\mathcal{R}_i\frac{{{{\bf{v}}_o}}}{{\sqrt {{N_v}{N_h}{N_f}} }}\\
& = \mathop {\lim }\limits_{{N_v},{N_h},{N_f} \to \infty } \int\limits_{(\theta ,\phi ,\tau ) \in {\Omega _i}} {\frac{{{{\left| {{{\bf{v}}_o}^H{\bf{v}}(\theta ,\phi ,\tau )} \right|}^2}}}{{{N_v}{N_h}{N_f}}}{p_i}} d\theta d\phi d\tau \nonumber \\
& \overset{{a}}{=} 0, \label{Eq:vo_nullSpace}
\end{align}
where $\overset{{a}}{=}$ is due to Lemma \ref{lemma:mutualOthogonal} and the fact that the term $| {{{\bf{v}}_o}^H{\bf{v}}(\theta ,\phi ,\tau )} |$ is a finite value. Eq. (\ref{Eq:vo_nullSpace}) shows that any $\vv_o$ with $(\theta_o, \phi_o, \tau_o) \not\in \Omega _i$ is in the null space of $\mathcal{R}_i$ when $N_v, N_h$, and $N_f$ are large. Since $\Omega_p \cap \Omega_q = \emptyset$, we may readily see that the signal space of $\mathcal{R}_p$ and $\mathcal{R}_q$ are asymptotically orthogonal to each other. More precisely, define the signal space of $\mathcal{R}_i$ as:
\begin{align}
    \operatorname{span}\{ \mathcal{R}_i\} \defi \operatorname{span}\{ \vu_n^{(i)}: n = 1, \cdots, r_i\}, i = p, q,
\end{align}
where $\vu_n^{(i)}$ is the $n$-th eigenvector corresponding to the $n$-th non-zero eigenvalue of $\mathcal{R}_i$. $r_i$ is the rank of $\mathcal{R}_i$. Then we have
\begin{equation}
\operatorname{span}\{ \mathcal{R}_p\} \perp \operatorname{span}\{ \mathcal{R}_q\}, \text{ as } {N_v},{N_h},{N_f} \to \infty.
\end{equation}
Next we will prove that $\operatorname{span}\{ \mathcal{R}_p\}$ and $\operatorname{span}\{ \mathcal{R}_q\}$ converge to certain mutually orthogonal DFT column spaces.
We define
\begin{align}
{\tilde{\bf{ U}}_i} = \left[ {\begin{array}{*{20}{c}}
{{\bf{u}}_1^{(i)}}&{{\bf{u}}_2^{(i)}}& \cdots &{{\bf{u}}_{{r_i}}^{(i)}}
\end{array}} \right], i = p,q.
\end{align}
Then, $\operatorname{span}\{ \mathcal{R}_i\}$ is also the column space of ${\tilde{\bf{ U}}_i}$.
We show that the signal spaces of the following three covariance matrices converge to certain column spaces of DFT submatrices.
\begin{align}
{{\cal R}_{v,i}} &= \int_{{\theta _i} - {\Delta _\theta }}^{{\theta _i} + {\Delta _\theta }} {\frac{1}{{2{\Delta _\theta }}}} {{\bf{a}}_v}(\theta ){{\bf{a}}_v}{(\theta )^H}d\theta \\
{{\cal R}_{h,i}} &= \int_{{\phi _i} - {\Delta _\phi }}^{{\phi _i} + {\Delta _\phi }} {\frac{1}{{2{\Delta _\phi }}}} {{\bf{a}}_h}(\theta ,\phi ){{\bf{a}}_h}{(\theta ,\phi )^H}d\phi \\
{{\cal R}_{f,i}} &= \int_{{\tau _i} - {\Delta _\tau }}^{{\tau _i} + {\Delta _\tau }} {\frac{1}{{2{\Delta _{{\tau }}}}}} {\bf{b}}(\tau ){\bf{b}}{(\tau )^H}d\tau
\end{align}
We look at ${\cal R}_{v,i}$. Without loss of generality, we assume $\cos(\theta_i + \Delta_\theta) < \cos(\theta_i - \Delta_\theta) < 0$.
Denote the set of indices for which the corresponding ``angular frequency" in the DFT matrix $\mW(N_v)$ belong to the range $[-D_v \cos(\theta_i - \Delta_\theta)/\lambda_0, -D_v \cos(\theta_i + \Delta_\theta)/\lambda_0]$
\begin{equation}
\mathcal{J}_{v,i} \defi \{n: \frac{n}{N_v} \in [-\frac{D_v \cos(\theta_i - \Delta_\theta)}{\lambda_0}, -\frac{D_v \cos(\theta_i + \Delta_\theta)}{\lambda_0}]\}
\end{equation}
Denote the DFT submatrix $\tilde{\mF}_{v,i}$ as the matrix containing the columns of $\mW(N_v)$ with indices in $\mathcal{J}_{v,i}$.

According to Corollary 1 of \cite{adhikary:13},
\begin{align}\label{Eq:UU_FF}
\mathop {\lim }\limits_{{N_v} \to \infty } \frac{1}{{{N_v}}}\left\| {{{\widetilde {\bf{U}}}_{v,i}}\widetilde {\bf{U}}_{v,i}^H - {{\widetilde {\bf{F}}}_{v,i}}\widetilde {\bf{F}}_{v,i}^H} \right\|_F^2 = 0,
\end{align}
where ${{\widetilde {\bf{U}}}_{v,i}}$ is composed of the eigenvectors corresponding to the non-zero eigenvalues of ${{\cal R}_{v,i}}$. The rank of ${{\cal R}_{v,i}}$ is $r_{v,i}$, which satisfy \cite{yin:13} \cite{adhikary:13}:
\begin{equation}\label{Eq:normRank}
\mathop {\lim }\limits_{{N_v} \to \infty } \frac{r_{v,i}}{N_v} = \frac{|\mathcal{J}_{v,i}|}{N_v},
\end{equation}
where $|\mathcal{J}_{v,i}|$ is the cardinality of $\mathcal{J}_{v,i}$. In other words
\begin{equation}
|\mathcal{J}_{v,i}| = r_{v,i} + o(N_v).
\end{equation}
From Eq. (\ref{Eq:UU_FF}) and Eq. (\ref{Eq:normRank}) we readily obtain:
\begin{align}
\mathop {\lim }\limits_{{N_v} \to \infty } \frac{1}{N_v} \left\{ r_{v,i} - \operatorname{tr}\{ {{\widetilde {\bf{U}}}_{v,i}}\widetilde {\bf{U}}_{v,i}^H  {{\widetilde {\bf{F}}}_{v,i}}\widetilde {\bf{F}}_{v,i}^H \} \right\} = 0,
\end{align}
or equivalently,
\begin{align}
\operatorname{tr}\{ {{\widetilde {\bf{U}}}_{v,i}}\widetilde {\bf{U}}_{v,i}^H  {{\widetilde {\bf{F}}}_{v,i}}\widetilde {\bf{F}}_{v,i}^H \} = r_{v,i} + o(N_v).
\end{align}
In a similar manner, we define the ranks, non-negligible eigenvectors, and the corresponding DFT submatrices of ${{\cal R}_{h,i}}$ and ${{\cal R}_{f,i}}$ as $r_{h,i}$, $r_{f,i}$, ${{\widetilde {\bf{U}}}_{h,i}}$, ${{\widetilde {\bf{U}}}_{f,i}}$, ${\widetilde {\bf{F}}}_{h,i}$, and ${\widetilde {\bf{F}}}_{f,i}$ respectively. The sets of DFT columns corresponding to ${\widetilde {\bf{F}}}_{h,i}$ and ${\widetilde {\bf{F}}}_{f,i}$ are denoted by $\mathcal{J}_{h,i}$ and $\mathcal{J}_{f,i}$
We can prove
\begin{align}
\operatorname{tr}\{ {{\widetilde {\bf{U}}}_{h,i}}\widetilde {\bf{U}}_{h,i}^H  {{\widetilde {\bf{F}}}_{h,i}}\widetilde {\bf{F}}_{h,i}^H \} & = r_{h,i} + o(N_h) \\
\operatorname{tr}\{ {{\widetilde {\bf{U}}}_{f,i}}\widetilde {\bf{U}}_{f,i}^H  {{\widetilde {\bf{F}}}_{f,i}}\widetilde {\bf{F}}_{f,i}^H \} & = r_{f,i} + o(N_f) \\
|\mathcal{J}_{h,i}| = r_{h,i} + o(N_h) \\
|\mathcal{J}_{f,i}| = r_{f,i} + o(N_f).
\end{align}
Now we examine the closeness of the column space of ${\widetilde {\bf{F}}}_f \otimes {\widetilde {\bf{F}}}_h \otimes {\widetilde {\bf{F}}}_v$ to the column space ${\widetilde {\bf{U}}}_f \otimes {\widetilde {\bf{U}}}_h \otimes {\widetilde {\bf{U}}}_v$. The difference between the two spaces is defined as
\begin{align}
\xi_i \defi & \left\| \left( {{{\widetilde {\bf{U}}}_{f,i}}\widetilde {\bf{U}}_{f,i}^H} \right) \otimes \left( {{{\widetilde {\bf{U}}}_{h,i}}\widetilde {\bf{U}}_{h,i}^H} \right) \otimes \left( {{{\widetilde {\bf{U}}}_{v,i}}\widetilde {\bf{U}}_{v,i}^H} \right) \right. \nonumber\\
& \left. - \left( {{{\widetilde {\bf{F}}}_{f,i}}\widetilde {\bf{F}}_{f,i}^H} \right) \otimes \left( {{{\widetilde {\bf{F}}}_{h,i}}\widetilde {\bf{F}}_{h,i}^H} \right) \otimes \left( {{{\widetilde {\bf{F}}}_{v,i}}\widetilde {\bf{F}}_{v,i}^H} \right) \right\|_F^2. \label{Eq:eigenApprErr}
\end{align}
For notational simplicity, we temporarily drop the subscript $i$. Then we may derive
\begin{small}
\begin{align*}
\xi_i &= {\mathop{\rm tr}\nolimits} \left\{ {{{\widetilde {\bf{U}}}_f}\widetilde {\bf{U}}_f^H{{\widetilde {\bf{U}}}_f}\widetilde {\bf{U}}_f^H} \right\}{\mathop{\rm tr}\nolimits} \left\{ {{{\widetilde {\bf{U}}}_h}\widetilde {\bf{U}}_h^H{{\widetilde {\bf{U}}}_h}\widetilde {\bf{U}}_h^H} \right\}{\mathop{\rm tr}\nolimits} \left\{ {{{\widetilde {\bf{U}}}_v}\widetilde {\bf{U}}_v^H{{\widetilde {\bf{U}}}_v}\widetilde {\bf{U}}_v^H} \right\}\\
& \quad + {\mathop{\rm tr}\nolimits} \left\{ {{{\widetilde {\bf{F}}}_f}\widetilde {\bf{F}}_f^H{{\widetilde {\bf{F}}}_f}\widetilde {\bf{F}}_f^H} \right\}{\mathop{\rm tr}\nolimits} \left\{ {{{\widetilde {\bf{F}}}_h}\widetilde {\bf{F}}_h^H{{\widetilde {\bf{F}}}_h}\widetilde {\bf{F}}_h^H} \right\}{\mathop{\rm tr}\nolimits} \left\{ {{{\widetilde {\bf{F}}}_v}\widetilde {\bf{F}}_v^H{{\widetilde {\bf{F}}}_v}\widetilde {\bf{F}}_v^H} \right\} \\
& \quad - {\mathop{\rm tr}\nolimits} \left\{ {{{\widetilde {\bf{U}}}_f}\widetilde {\bf{U}}_f^H{{\widetilde {\bf{F}}}_f}\widetilde {\bf{F}}_f^H} \right\}{\mathop{\rm tr}\nolimits} \left\{ {{{\widetilde {\bf{U}}}_h}\widetilde {\bf{U}}_h^H{{\widetilde {\bf{F}}}_h}\widetilde {\bf{F}}_h^H} \right\}{\mathop{\rm tr}\nolimits} \left\{ {{{\widetilde {\bf{U}}}_v}\widetilde {\bf{U}}_v^H{{\widetilde {\bf{F}}}_v}\widetilde {\bf{F}}_v^H} \right\} \\
& \quad - {\mathop{\rm tr}\nolimits} \left\{ {{{\widetilde {\bf{F}}}_f}\widetilde {\bf{F}}_f^H{{\widetilde {\bf{U}}}_f}\widetilde {\bf{U}}_f^H} \right\}{\mathop{\rm tr}\nolimits} \left\{ {{{\widetilde {\bf{F}}}_h}\widetilde {\bf{F}}_h^H{{\widetilde {\bf{U}}}_h}\widetilde {\bf{U}}_h^H} \right\}{\mathop{\rm tr}\nolimits} \left\{ {{{\widetilde {\bf{F}}}_v}\widetilde {\bf{F}}_v^H{{\widetilde {\bf{U}}}_v}\widetilde {\bf{U}}_v^H} \right\} \\
& = {r_f}{r_h}{r_v} + \left( {{r_f} + o({N_f})} \right)\left( {{r_h} + o({N_h})} \right)\left( {{r_v} + o({N_v})} \right) \\
& \quad - 2\left( {{r_f} + o({N_f})} \right)\left( {{r_h} + o({N_h})} \right)\left( {{r_v} + o({N_v})} \right).
\end{align*}
\end{small}
Then, it is clear that
\begin{equation}\label{Eq:eigenApproxErrIs0}
\mathop {\lim }\limits_{{N_v},{N_h},{N_f} \to \infty } \frac{1}{{{N_v}{N_h}{N_f}}} \xi_i = 0.
\end{equation}
Eq. (\ref{Eq:eigenApproxErrIs0}) indicates that when ${N_v},{N_h},{N_f}$ are large, the column space of ${\bar {\bf{F}}}_{i}$ converges to ${\bar {\bf{U}}}_{i}$, where ${\bar {\bf{F}}}_{i} \defi {\widetilde {\bf{F}}}_{f, i} \otimes {\widetilde {\bf{F}}}_{h, i} \otimes {\widetilde {\bf{F}}}_{v, i}$ and ${\bar {\bf{U}}}_{i} \defi {\widetilde {\bf{U}}}_{f, i} \otimes {\widetilde {\bf{U}}}_{h, i} \otimes {\widetilde {\bf{U}}}_{v, i}$. Since $\operatorname{span}\{ \mathcal{R}_i\}$ is equivalent to the column space of ${\bar {\bf{U}}}_{i}$, according to the orthogonality between $\operatorname{span}\{ \mathcal{R}_p\}$ and $\operatorname{span}\{ \mathcal{R}_q\}$, the column spaces of ${\bar {\bf{F}}}_{p}$ and ${\bar {\bf{F}}}_{q}$ are also asymptotically orthogonal. In other words, define the column space of ${\bar {\bf{F}}}_{i}$:
\begin{equation}
{\bar {\cal B}_i} \defi \text{span} \{\vf_{i,n}: n = 1, \cdots, M_i \}, i = p, q,
\end{equation}
where $\vf_{i,n}$ is the $n$-th column of ${\bar {\bf{F}}}_{i}$ and $M_i$ is the number of columns of ${\bar {\bf{F}}}_{i}$. Then
\begin{equation}
{\bar {\cal B}_p} \bot {\bar {\cal B}_q} \text{ when } {N_v},{N_h},{N_f} \rightarrow \infty.
\end{equation}
As ${\bar {\bf{F}}}_{i}$ is a submatrix of the unitary matrix $\mS$ as in Eq. (\ref{Eq:Smat}), ${\bar {\bf{F}}}_{p}$ and ${\bar {\bf{F}}}_{q}$ have no shared columns of $\mS$ when ${N_v},{N_h},{N_f} \rightarrow \infty$. Furthermore, since $(\theta _{p}, {\phi _{p}}, \tau_p) \in \Omega_p$ and $(\theta _{q}, {\phi _{q}}, \tau_q) \in \Omega_q$, it follows that ${\cal B}_p   \subseteq {\bar {\cal B}_p}$ and ${\cal B}_q   \subseteq {\bar {\cal B}_q}$. Therefore we have
\begin{equation}
\mathcal{B}_p \bot \mathcal{B}_q \text{ as } {N_v, N_h, N_f \rightarrow \infty},
\end{equation}
which proves Lemma \ref{lemma:orthLinearSpace}.
\end{proof}

\subsection{Proof of Lemma \ref{lemma:pronyOrder1}:}\label{proof:pronyOrder1}
\begin{proof}
Since only two neighboring samples $y(m-1)$ and $y(m)$ are available, the order of the linear prediction is 1. We may obtain an estimate of the prediction coefficient $p_0$ according to Prony's method in Sec. \ref{sec:pronyReview} by solving the linear equation
\begin{equation}
y(m-1) p_0 = - y(m),
\end{equation}
where the solution is given by $\hat{p}_0 = - e^{j 2 \pi f}$. Now applying the linear prediction $\hat{y}(n+1) = - \hat{p}_0 y(n), \forall n \ge m$, we may obtain
\begin{align}
\hat{y}(m+N_d) & = (- \hat{p}_0)^{N_d} y(m) = e^{j 2 \pi f N_d } y(m) \\
& = \beta e^{j2\pi f m} e^{j 2 \pi f N_d } = {y}(m+N_d).
\end{align}
Thus, Lemma \ref{lemma:pronyOrder1} is proved.
\end{proof}

\subsection{Proof of Theorem \ref{theoNoiseFreeEstimationErr}:}\label{proof:theoNoiseFreeEstimationErr}
\begin{proof}
For a certain path $p$, define a submatrix $\tilde \mS_p \in \mathbb{C}^{N_t N_f \times S_p}$ of $\mS$:
\begin{equation}
\tilde \mS_p = \left[ {\begin{array}{*{20}{c}}
{{{\bf{s}}_{p,1}}}&{{{\bf{s}}_{p,2}}}& \cdots &{{{\bf{s}}_{p,{S_p}}}}
\end{array}} \right],
\end{equation}
where all columns of $\tilde \mS_p$ are chosen from $\mS$ with indices belonging to the set $\mathcal{M}_p$, which is defined in Eq. (\ref{Eq:M_p}). According to the definition of $\mathcal{M}_p$, the generalized steering vector $\vv_p$ is in the column space of $\tilde \mS_p$. Thus
\begin{equation}
 \lim_{N_v, N_h, N_f \rightarrow \infty} \frac{\left\|\tilde \mS_p^H  \vv_p \right\|^2}{N_v N_h N_f} = 1.
\end{equation}
We now consider $\vg_u(t)$ in Eq. (\ref{Eq:Gt}). Due to Lemma \ref{lemma:orthLinearSpace} and the condition Eq. (\ref{Eq:noiseFreeCond}), we may group the non-vanishing rows of $\vg_u(t)$ into $P$ set, with each set corresponding to a certain path. Notice that the sample error does not affect the selection of the non-vanishing rows here since according to Eq. (\ref{Eq:accurateSamples}) it converges to zero. For a certain $n \in \mathcal{M}_p$, we may derive
\begin{small}
\begin{align}
& \mathop {\lim }\limits_{{N_v},{N_h},{N_f} \to \infty } \frac{{{\bf{s}}_{{n}}^H{\tilde{\bm{\hbar}} _u}(t)}}{{\sqrt {{N_v}{N_h}{N_f}} }} = \mathop {\lim }\limits_{{N_v},{N_h},{N_f} \to \infty } \frac{{{\bf{s}}_{{n}}^H{\bm{\hbar} _u}(t)}}{{\sqrt {{N_v}{N_h}{N_f}} }} \\
& = \mathop {\lim }\limits_{{N_v},{N_h},{N_f} \to \infty } \frac{{{g_{u,n}}(t)}}{{\sqrt {{N_v}{N_h}{N_f}} }}  = \mathop {\lim }\limits_{{N_v},{N_h},{N_f} \to \infty } \frac{{\sum\limits_{i = 1}^P {{c_{u,i}}} (t){\rm{ }}{\bf{s}}_{{n}}^H{{\bf{v}}_p}}}{{\sqrt {{N_v}{N_h}{N_f}} }}  \nonumber \\
&= \mathop {\lim }\limits_{{N_v},{N_h},{N_f} \to \infty } \frac{{{\rm{ }}{c_{u,p}}(t){\bf{s}}_{{n}}^H{{\bf{v}}_p}}}{{\sqrt {{N_v}{N_h}{N_f}} }} = \mathop {\lim }\limits_{{N_v},{N_h},{N_f} \to \infty } {{\eta _{p,u,n}} e^{j{\omega _p}t}}, \label{Eq:ita_pun}
\end{align}
\end{small}
where ${g_{u,n}}(t)$ is the $n$-th row of $\vg_u(t)$ and
\begin{equation}
{\eta _{p,u,n}} = \frac{{{\bf{s}}_n^H{{\bf{v}}_p}{\beta _p}{e^{\frac{{j2\pi \hat r_{{\rm{rx}},p}^T{{\bar d}_{{\rm{rx}},u}}}}{{{\lambda _0}}}}}}}{{\sqrt {{N_v}{N_h}{N_f}} }}.
\end{equation}
We can see that ${\eta _{p,u,n}}$ is time invariant and is not affected by the vanishing sample error. Moreover, since $|{\eta _{p,u,n}}| \leq {\beta _p}$, Eq. (\ref{Eq:ita_pun}) converges to an exponential signal with only one pole frequency, which can be predicted without error using Prony's method even with only two neighboring samples. The same conclusion holds for the rows in all $P$ sets, while the other rows converge to zero when normalized by ${\sqrt {{N_v}{N_h}{N_f}}}$. Therefore,
\begin{align}
\mathop {\lim }\limits_{{N_v},{N_h},{N_f} \to \infty } \frac{\hat{\vg}_u{(t_{L+N_d})} - \vg_u{(t_{L+N_d})}}{{\sqrt {{N_v}{N_h}{N_f}} }} = 0,
\end{align}
where $\vg_u{(t_{L+N_d})} = \mS^H \bm{\hbar}_u(t_{L+N_d})$ and $\hat{\vg}_u{(t_{L+N_d})}$ is the prediction using the PAD algorithm.
Notice that
\begin{align}
\mathop {\lim }\limits_{{N_v},{N_h},{N_f} \to \infty } \frac{{\left\|{\bm{\hbar}}_u(t_{L+N_d})\right\|_2^2}}{{{N_v}{N_h}{N_f}}} = \sum\limits_{p = 1}^P {\beta _p^2},
\end{align}
when condition Eq. (\ref{Eq:noiseFreeCond}) is fulfilled. We may further derive
\begin{align}
& \lim_{N_v, N_h, N_f \rightarrow \infty} \frac{\left\|{\hat{\bm{\hbar}}_u(t_{L+N_d}) - \bm{\hbar}}_u(t_{L+N_d})\right\|_2^2}{\left\|{\bm{\hbar}}_u(t_{L+N_d})\right\|_2^2} \\
& = \mathop {\lim }\limits_{{N_v},{N_h},{N_f} \to \infty } \frac{ \left\|\hat{\vg}_u{(t_{L+N_d})} - \vg_u{(t_{L+N_d})}\right\|_2^2}{{N_v}{N_h}{N_f} \sum\limits_{p = 1}^P {\beta _p^2}} \\
& = 0,
\end{align}
which proves Theorem \ref{theoNoiseFreeEstimationErr}.
\end{proof}
%

\bibliographystyle{IEEEtran}
\bibliography{bib/allCitations}
\bibliographystyle{IEEEtran}
\begin{IEEEbiography}[{\includegraphics[width=1in,height=1.25in,clip,keepaspectratio]{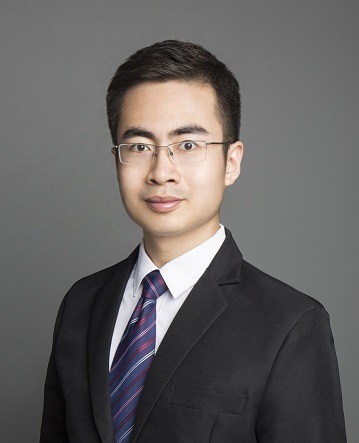}}]
{Haifan Yin} received the Ph.D. degree from T\'el\'ecom ParisTech in 2015. He received the B.Sc. degree in Electrical and Electronic Engineering and the M.Sc. degree in Electronics and Information Engineering from Huazhong University of Science and Technology, Wuhan, China, in 2009 and 2012 respectively. From 2009 to 2011, he has been with Wuhan National Laboratory for Optoelectronics, China, working on the implementation of TD-LTE systems as an R\&D engineer.
From 2016 to 2017, he has been a DSP engineer in Sequans Communications - an IoT chipmaker based in Paris, France. From 2017 to 2019, he has been a senior research engineer working on 5G standardization in Shanghai Huawei Technologies Co., Ltd., where he made substantial contributions to 5G standards, particularly the 5G codebooks. Since May 2019, he has joined the School of Electronic Information and Communications at Huazhong University of Science and Technology as a full professor.
His current research interests include 5G and 6G networks, signal processing, machine learning, and massive MIMO systems.

\end{IEEEbiography}
%
\begin{IEEEbiography}[{\includegraphics[width=1in,height=1.25in,clip,keepaspectratio]{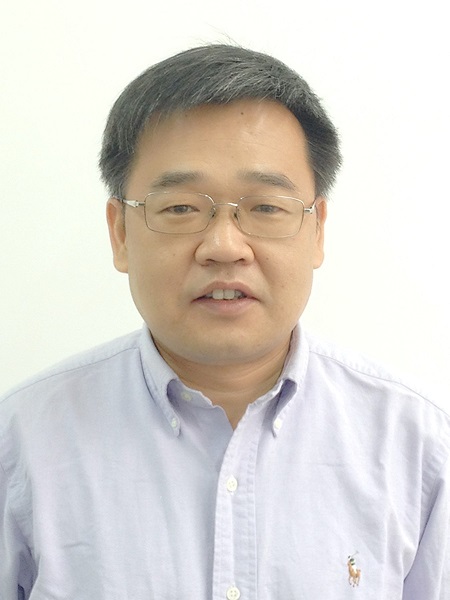}}]
{Haiquan Wang} (M'03-SM'19) received the M.S. degree in Nankai University, China, in 1989, and the Ph.D. degree in Kyoto University, Japan, in 1997, both in mathematics, and the Ph.D. degree in University of Delaware, Newark, in 2005, in electrical engineering. From 1997 to 1998, he was a Postdoctoral Researcher with the Department of Mathematics, Kyoto University. From 1998 to 2001, he was a Lecturer (part-time) with Ritsumei University, Japan. From 2001 to 2002, he was a Visiting Scholar with the Department of Electrical and Computer Engineering, University of Delaware, Newark. From 2005 to 2008, he was a Postdoctoral Researcher with the Department of Electrical and Computer Engineering, University of Waterloo, Canada. He joined the College of Communications Engineering, Hangzhou Dianzi University, Hangzhou, China, in July 2008 as a faculty member. His research interests include wireless communications, digital signal processing, and information theory.
\end{IEEEbiography}
%
\begin{IEEEbiography}[{\includegraphics[width=1in,height=1.25in,clip,keepaspectratio]{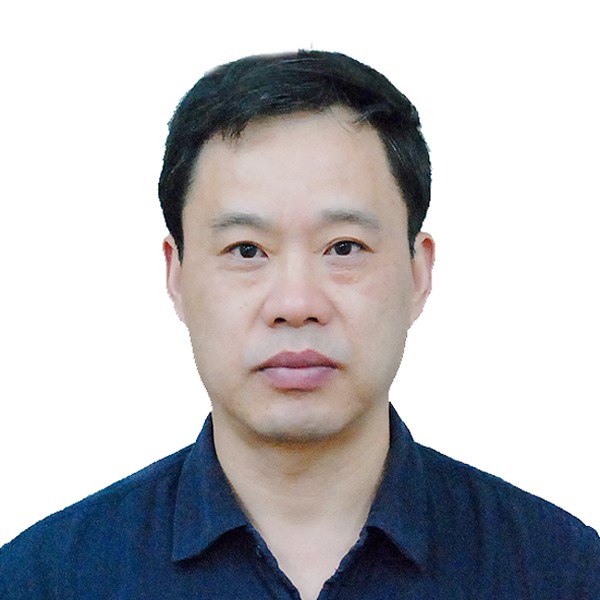}}]
{Yingzhuang Liu} is Professor of Huazhong University of Science and Technology (HUST). His main research field is broadband wireless communications, include 5G/6G, WLAN system, etc. From 2000 to 2001, he was a postdoctoral researcher in Paris University XI. Since 2003, he has presided more than 10 national key projects, published more than 100 papers and has more than 50 patents in the field of broadband wireless communication. He is now the group leader of broadband wireless research of HUST.

\end{IEEEbiography}

\begin{IEEEbiography}[{\includegraphics[width=1in,height=1.25in,clip,keepaspectratio]{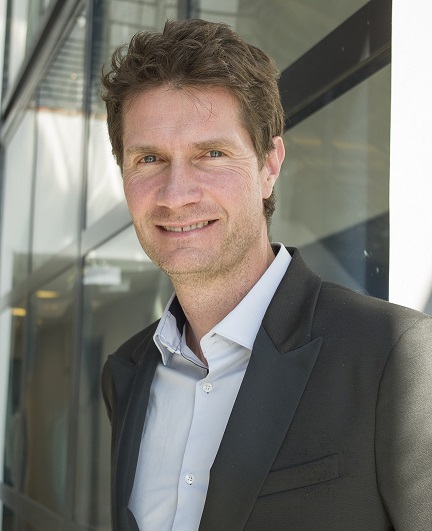}}]
{David Gesbert} (IEEE Fellow) is Professor and Head of the Communication Systems Department, EURECOM. He obtained the Ph.D. degree from Ecole Nationale Superieure des Telecommunications, France, in 1997. From 1997 to 1999 he has been with the Information Systems Laboratory, Stanford University. He was then a founding engineer of Iospan Wireless Inc, a Stanford spin off pioneering MIMO-OFDM (now Intel). Before joining EURECOM in 2004, he has been with the Department of Informatics, University of Oslo as an adjunct professor. D. Gesbert has published about 300 papers and 25 patents, some of them winning 2019 ICC Best Paper Award, 2015 IEEE Best Tutorial Paper Award (Communications Society), 2012 SPS Signal Processing Magazine Best Paper Award, 2004 IEEE Best Tutorial Paper Award (Communications Society), 2005 Young Author Best Paper Award for Signal Proc. Society journals, and paper awards at conferences 2011 IEEE SPAWC, 2004 ACM MSWiM. He has been a Technical Program Co-chair for ICC2017. He was named a Thomson-Reuters Highly Cited Researchers in Computer Science. Since 2015, he holds the ERC Advanced grant "PERFUME" on the topic of smart device Communications in future wireless networks. He is a Board member for the OpenAirInterface (OAI) Software Alliance. Since early 2019, he heads the Huawei-funded Chair on Advanced Wireless Systems Towards 6G Networks. He sits on the Advisory Board of HUAWEI European Research Institute.

\end{IEEEbiography}



\end{document}